\documentclass{article}
\usepackage{epsfig}
\usepackage{amsthm}
\usepackage{amsmath}
\usepackage{amssymb}
\usepackage{amsbsy}
\usepackage{amscd}
\usepackage{color}
\usepackage{graphicx}
\usepackage{verbatim}
\newcommand{\abs}[1]{\left\lvert #1 \right\rvert}\def\wt{\widetilde}
\DeclareMathOperator{\sgn}{sgn}

\def\wh{\widehat}

\def\ds{\displaystyle}
\def\res{\mathop{\mathrm{res}}\limits_}

\makeatletter
\@addtoreset{equation}{section}
\makeatother
\def\I{{\cal{I}}}
\def\tr{\mathrm {Tr}}
\def\le{\left}
\def\ri{\right}
\def\YY{\mathbb Y}

\def\bc{\begin{corollary}}
\def\ec{\end{corollary}}
\def\&{&{\hskip -20pt}}

\def \p{\mathbf p}
\def\q{\mathbf q}
\def\u{\mathbf u}

\def\br{\begin{remark}\rm\small}
\def\1{{\bf 1}}
\def\er{\end{remark}}
\def\bt{\begin{theorem}}
\def\et{\end{theorem}}

\def\bx{\begin{examp}}
\def\ex{\end{examp}}
\def\bd{\begin{definition}}
\def\ed{\end{definition}}
\def\bp{\begin{proposition}\rm}
\def\bl{\begin{lemma}\em}
\def\el{\end{lemma}}
\def\ep{\end{proposition}}
\def\bea{\begin{eqnarray}}
\def\eea{\end{eqnarray}}

\def\C{{\mathbb C}}
\def\R{{\mathbb R}}
\def\N{{\mathbb N}}
\def\Z{{\mathbb Z}}

\newtheorem{theorem}{Theorem}[section]
\newtheorem{examp}{Example}[section]
\newtheorem{coroll}{Corollary}[section]
\newtheorem{examps}{Examples}[section]

\newtheorem{lemma}{Lemma}[section]
\newtheorem{remark}{Remark}[section]
\newtheorem{remarks}[remark]{Remarks}
\newtheorem{proposition}{Proposition}[section] 
\newtheorem{definition}{Definition}[section]
\def\br{\begin{remark}}
\def\er{\end{remark}}
\def\bt{\begin{theorem}}
\def\et{\end{theorem}}
\def\bc{\begin{coroll}}
\def\ec{\end{coroll}}
\def\brs{\begin{remarks} \rm\
\begin{enumerate}}
\def\ers{\end{enumerate}\end{remarks}}
\def\bl{\begin{lemma}}
\def\el{\end{lemma}}
\def\bxs{\begin{examps}. \rm\begin{enumerate}}
\def\exs{\end{enumerate}\end{examps}}
\def\bd{\begin{definition}}
\def\ed{\end{definition}}
\def\bp{\begin{proposition}}
\def\ep{\end{proposition}}
\def\be{\begin{equation}}
\def\ee{\end{equation}}
\def\bew{\begin{equation*}}
\def\eew{\end{equation*}}

\def\X {{\bf X}}

\def\Y{{\bf Y}}

\def\d{{\rm d}}
\def\bea{\begin{eqnarray}}
\def\eea{\end{eqnarray}}
\def\beas{\begin{eqnarray*}}
\def\eeas{\end{eqnarray*}}

\def\iint{\int\!\!\!\!\int}
\def\C{{\mathbb C}}

\def\R{{\mathbb R}}
\def\N{{\mathbb N}}

\def\Z{{\mathbb Z}}

\date{}
\begin{document}
%                                                                  
%\fontfamily{cmss}
%\fontsize{11pt}{15pt}
%\selectfont

\baselineskip 16pt plus 1pt minus 1pt
%\begin{titlepage}
\begin{flushright}
%CRM-???? (2007)\\
%nlin.SI/05xxxxxx
\end{flushright}
\vspace{0.2cm}
\begin{center}
\begin{Large}
%\fontfamily{cmss}
%\fontsize{17pt}{27pt}
%\selectfont
\textbf{ Cauchy Biorthogonal Polynomials}
\end{Large}\\
\bigskip
\begin{large} {M.
Bertola $^{\dagger\ddagger}$ \footnote{Work supported in part by the Natural
    Sciences and Engineering Research Council of Canada (NSERC),
    Grant. No. 261229-03 and by the Fonds FCAR du
    Qu\'ebec No. 88353.},  M. Gekhtman~$^a$ 
\footnote{Work supported in part by NSF Grant DMD-0400484.}, J. Szmigielski~$^b$ \footnote{Work supported in part by the Natural
    Sciences and Engineering Research Council of Canada (NSERC),
    Grant. No. 138591-04}}
\end{large}
\\
\bigskip
\begin{small}
$^{\dagger}$ {\it Centre de recherches math\'ematiques,
Universit\'e de Montr\'eal\\ C.~P.~6128, succ. centre ville, Montr\'eal,
Qu\'ebec, Canada H3C 3J7 \\
~~E-mail: bertola@crm.umontreal.ca}
\smallskip

$^{\ddagger}$ {\it Department of Mathematics and
Statistics, Concordia University\\ 1455 de Maisonneuve W., Montr\'eal, Qu\'ebec,
Canada H3G 1M8} \\
\smallskip
~$^a$ {\it Department of Mathematics
255 Hurley Hall, Notre Dame, IN 46556-4618, USA\\
~~E-mail: Michael.Gekhtman.1@nd.edu}

\smallskip
~$^b$ {\it Department of Mathematics and Statistics, University of Saskatchewan\\ 106 Wiggins Road, Saskatoon, Saskatchewan, S7N 5E6, Canada\\
~~E-mail: szmigiel@math.usask.ca}
\end{small}
\end{center}
\bigskip
%%%%%%%%%%%%%%%%  Abstract  %%%%%%%%%%%%%%%%
%\begin{center}{\bf Abstract}\\
%\end{center}
\begin{abstract}
The paper investigates the properties of certain biorthogonal polynomials appearing in a specific simultaneous Hermite-Pad\'e\ approximation scheme. Associated to any totally positive kernel and a pair of positive measures on the positive axis we define biorthogonal polynomials  and prove that their zeroes are simple and positive. We then specialize the kernel to the Cauchy kernel $\frac 1{x+y}$ and show that the ensuing biorthogonal polynomials solve a four-term recurrence relation, have relevant Christoffel-Darboux generalized formul\ae\, and their zeroes are interlaced. In addition, these polynomial 
solve a combination of Hermite-Pad\'{e} approximation problems to a Nikishin system 
of order $2$.  
The motivation arises from two distant areas; on one side, in the study of the inverse spectral problem for the peakon solution of the Degasperis-Procesi equation; on the other side, from a random matrix model involving two positive definite random Hermitian matrices. Finally, we show how to characterize these polynomials in term of a Riemann--Hilbert problem.

% Peakons are non-smooth soliton solutions appearing in certain
% nonlinear partial differential equations of current interest, like the 
% Camassa-Holm equation and the Degasperis-Procesi equation. In the
% former case the construction of peakons uses Stieltjes continued fractions and discrete orthogonal 
% polynomials.  In the case of the Degasperis-Procesi equation the peakon construction leads 
% to a specific type of 
% biorthogonal polynomials associated to the Cauchy kernel $\frac{1}{x+y}$ whose general theory is developed in this paper.  Fundamental properties of these polynomials
% are proved: their zeros are interlaced, they satisfy four-term
% recurrence relations, generalized Christoffel-Darboux identities and solve a Hermite-Pad\'{e} 
% type approximation problem as well as 
% they admit a characterization in terms of a $3$ by $3$ matrix Riemann-Hilbert problem. 
%\end{center}
\end{abstract}
\medskip
%\begin{small}
%\end{small}
\bigskip

%\end{titlepage}
\tableofcontents

 \section{Introduction and motivations}
 This paper mainly deals with a class of biorthogonal polynomials $ \{p_n(x)\}_\N,\{ q_n(y)\}_\N$ of degree $n$ satisfying the 
 biorthogonality relations
 \be
 \int_{\R_+} \int_{\R_+} p_n(x) q_m(y) \frac {\d \alpha(x)\d\beta(y)}{x+y} = \delta_{mn}, \label{CBOPs}
 \ee
 where $\d\alpha, \d\beta$ are positive measures supported on $\R_+$ with finite bimoments. These polynomials will be introduced in Sec. \ref{SecPosKer} in a more general context of polynomials associated to general {\it totally positive kernels} (Def. \ref{posker}) with which they share some general properties in regard to their zeroes.
 
 While these properties are interesting in their own right, we wish to put the work in a more general context and explain the two main motivations behind it. They fall within two different and rather distant areas of mathematics :  peakon solutions to nonlinear PDEs and  Random Matrix theory.

\paragraph{ Peakons for the Degasperis-Procesi equation.}
%Polynomials of this type were first used in \cite{ls-cubicstring} to construct a special class of solutions (peakons)  to the weakly dispersive equation:
% \begin{equation}
%  \label{eq:DP}
%  u_t - u_{xxt} + 4 u u_x = 3 u_x u_{xx} + u u_{xxx}, \qquad (x,t)\in \mathbb{R}^2.  
%\end{equation}

%It might be of interest to the reader to know the context in which one studies equations of this type and why the theory of orthogonal polynomials and Pad\'{e} approximations played a decisive role in finding the peakon solutions.  
%The recent resurgence of research of this class of equations started with the equation: 
In the early 1990's, Camassa and Holm \cite{ch} introduced the (CH) equation to model (weakly) dispersive shallow wave propagation.
%\begin{equation}
%  \label{eq:CH}
%  u_t - u_{xxt} + 3u u_x = 2u_x u_{xx} + u u_{xxx},
%  \quad (x,t) \in \R^2,
%\end{equation} 
%which was proposed in the early 1990s by Camassa and Holm \cite{ch} (CH equation) as a model shallow water wave equation.  
More generally, the CH equation belongs to the so-called b-family of PDEs
\begin{equation}
  \label{eq:family}
  u_t - u_{xxt} + (b+1) u u_x = b u_x u_{xx} + u u_{xxx}, 
  \quad (x,t) \in \R^2, \quad b\in \R, 
\end{equation} 
Two cases, $b=2$ and $b=3$ within this family are now known to be integrable: the case $b=2$ is the original CH equation whereas the case $b=3$ is the Degasperis-Procesi \cite{dp} (DP) equation, which is more directly related to the present paper.

In all cases the b-family admits weak (distributional) solutions of the form: 
\begin{equation} \label{eq:peakonansatz}
  u(x,t) = \sum_{i=1}^n m_i(t) \, e^{-\abs{x-x_i(t)}},
\end{equation}
if and only if the positions $x_i(t)$ and the heights $m_i(t)$
satisfy the system of nonlinear ODEs:
\begin{equation}
  \label{eq:CH-peakonODE}
   \dot{x}_k = \sum_{i=1}^n m_i e^{-\abs{x_k-x_i}},
    \qquad 
    \dot{m}_k =(b-1) \sum_{i=1}^n m_k m_i \sgn(x_k-x_i) \, e^{-\abs{x_k-x_i}},
\end{equation}
for $k=1,\ldots,n$.
The non-smooth character of the solution manifests itself by the presence of sharp peaks 
at $\{x_k\}$, hence the name {\sl{peakons}}. 
For the CH equation the peakons solution were studied in \cite{bss-stieltjes, bss-moment}, while for the DP equation in \cite{ls-invprob, ls-cubicstring}; in both cases the solution is related to the {\bf isospectral evolution} of an associated linear boundary-value problem
\bea
\begin{array}{c|c}
b=2\  (CH) & b=3\ (DP)\\[10pt]
\hline\\
-\phi ''(\xi, z)  =z g (\xi) \phi (\xi, z) & 
-\phi'''(\xi,z)= z g(\xi) \phi(\xi,z) \\[7pt]
\phi(-1)= \phi(1)=0&
\phi(-1) = \phi'(-1)  = \phi(1)=0 
 \end{array}
\eea
The variables $x,\xi$ and the quantities $m, g, u$ are related by
\be
\xi=\tanh\le(\frac {x}{b-1}\ri)\ ,\qquad
g(\xi) = \le( \frac {1-\xi^2}{2}\ri)^{-b} m(x)\ ,\qquad
m(x,t) = u(x,t)-u_{xx}(x,t)
\ee
 Because of the similarity to the equation of an inhomogeneous classical string (after a separation of variables) we refer to the two linear ODEs as the {\it quadratic} and {\it cubic} string, respectively.
The case of peakons corresponds to the choice 
\be
 m(x,t) = 2\sum_{j=1}^n \delta(x-x_i(t)) m_i(t)\label{diraccomb}
 \ee
The remarkable fact is that in both cases the associated spectral problems have  a finite {\bf positive} spectrum; this is not so surprising in the case of the quadratic string which is a self-adjoint problem, but it is quite unexpected for the cubic string, since the problem is not self-adjoint and there is no {\it a priori} reason for the spectrum to even be real \cite{ls-cubicstring}.  

As it is natural within the Lax approach to integrable PDEs, the {\it spectral map} linearizes the evolution of the isospectral evolution: if $\{z_j\}$ are the eigenvalues of the respective {boundary value } problems and one introduces the appropriate spectral residues 
\be
b_j := \res{z=z_j}\frac{W(z)}{z} \d z\ ,\ \ \ \ \ W(z):=  \frac {\phi_\xi(1,z)}{\phi(1,z)}
\ee
then one can show \cite{ls-invprob} that the evolution linearizes as follows (with the dot representing the time evolution)
\bea
\dot z_k=0\ ,\qquad \frac{\dot b_k}{b_k} = \frac 1{z_k}
\eea
Since this is not the main focus of the paper, we are deliberately glossing over several interesting points; {the interested reader is referred to \cite{ls-cubicstring}   and our recent work \cite{Paper1mini} for further details. In short, the solution method for the DP equation  can by illustrated by the diagram}
$$\begin{CD}
\{ x_k(0), m_k(0) \}_{k=1}^n @> \phantom{\text{aa}} \text{ spectral map}\phantom{\text{aa}}>> \{z_k, b_k\}\\
@VV \text{DP flow}V@VV \text{evolution of the extended spectral data}V\\
\{ x_k(t), m_k(t) \}_{k=1}^n @<\text{inverse spectral map}< < \{\ds z_k(t)= z_k\qquad \qquad\qquad\ \ \ \atop \ds b_k(t) = b_k(0) \exp( t/z_k) \}
\end{CD}$$ 
In the inverse spectral map resides the r\^ole of the biorthogonal polynomials to be studied here, as we briefly sketch below. 
The inverse problem for the ordinary string with finitely many point masses 
is solved by the method of continued fractions of Stieltjes' type as was pointed out 
by M.G. Krein (\cite{gantmacher-krein}).  The inverse problem for the cubic string 
with finitely many masses is solved with the help of the following 
simultaneous Hermite-Pad\'e type approximation (\cite{ls-cubicstring})
\begin{definition}[Pad\'{e}-like approximation problem]
\label{def:pade}
Let $d\mu(x)$ denote the spectral measure associated with the cubic string boundary value problem and $\frac{W(z)}{z}=\int\frac{1}{z-x}d\mu(x)$, $\frac{Z}{z}=\iint \frac{x}{z-x}\frac{1}{x+y}d\mu(x)d\mu(y) $ denote the Weyl functions introduced in \cite{ls-cubicstring}.  Then, given an integer $1 \leq k \leq n$,
we seek three polynomials
$(Q,P,\widehat{P})$ of degree $k-1$
satisfying the following conditions:
\begin{enumerate}
\item[] {\bf [Approximation]}:
$\ds 
W=\frac{P}{Q}+O\left(\frac{1}{z^{k-1}}\right),
\qquad
Z=\frac{\widehat{P}}{Q}+O\left(\frac{1}{z^{k-1}}\right)
\qquad
(z\to\infty).
$
\item[]{\bf  [Symmetry]}:
$
\ds Z^* \, Q + W^* \, P + \widehat{P}
=O\left(\frac{1}{z^k}\right) \  (z\to\infty)
$ with $W^*(z)=-W(-z)$, $Z^*(z)=Z(-z)$.

\item []{\bf [Normalization]}:
$\ds 
P(0)=1,
\qquad
\widehat{P}(0)=0.
$
\end{enumerate}
\end{definition}

This approximation problem has a unique solution (\cite{ls-cubicstring}) which, in turn, 
is used to solve the inverse problem for the cubic string.  We point out that it is here in 
this approximation problem that the Cauchy kernel $\frac{1}{x+y}$ makes its, somewhat unexpected, appearance through the spectral representation of the second Weyl function.

\paragraph{Random Matrix Theory}
The other source of our interest in biorthogonal polynomials comes from random matrix theory.  It is well known \cite{MehtaBook} that the Hermitean matrix model is
intimately related to (in fact, solved by) orthogonal polynomials (OPs). Not
so much is known about the role of biorthogonal polynomials (BOPs).  
However, certain
biorthogonal polynomials somewhat similar to the ones in the 
present paper appear prominently in the analysis of ``the''
two--matrix model after reduction to the spectrum of eigenvalues \cite
{BEH_dualityCMP, BEHDuality, BEH_diffCMP, KenNick}; in that case the pairing is of the form
\be \int\int p_n(x) q_m(y) {\rm e}^{-xy} \d \alpha (x)
\d\beta(y)=\delta_{mn},  \ \label {IZBOPS} \ee 
and the associated biorthogonal polynomials are sometimes called the Itzykson--Zuber BOPs, in short, the IZBOPs.  
 
Several algebraic structural properties of these
polynomials and their recurrence relation (both multiplicative and
differential) have been thoroughly analyzed in the previously cited
papers for densities of the form $\d\alpha(x) = {\rm e}^{-V_1(x)}\d x,
\ \d \beta(y) = {\rm e}^{-V_2(y)}\d y$ for {\it polynomials potentials}
$V_1(x), \ V_2(y)$ and for potentials with rational derivative (and
hard--edges) in \cite{Bertosemiclass}.

We recall that while ordinary OPs satisfy a multiplicative
three--term recurrence relation, the BOPs defined by \eqref{IZBOPS}
solve a longer recurrence relation of length related to the degree of
the differential $\d V_j(x)$ over the Riemann sphere
\cite{Bertosemiclass}; a direct (although not immediate)
consequence of the finiteness of the recurrence relation is the fact
that these BOPs (and certain integral transforms of them) are
characterized by a Riemann--Hilbert problem for a matrix of size equal
to the length of the recurrence relation (minus one).
The BOPs introduced in this paper share all these features, although in some respects they are closer to the ordinary orthogonal polynomials than to the IZBOPs. 

The relevant two--matrix model our polynomials are related to was introduced in \cite{Paper2}.  We now give a brief summary of that work.  
Consider the set of pairs $\mathcal H_+^{(2)}:=\{ (M_1,M_2)\} $ of Hermitean {\it positive-definite} matrices endowed with the ($U(N)$--invariant) Lebesgue measure denoted by $\d M_1\d M_2$. 
Define then the probability measure on this space by the formula:
\be
\d \mu(M_1,M_2) = \frac 1{\mathcal Z_N^{(2)}} \frac {\alpha'(M_1)\beta'(M_2) \d M_1 \d M_2}{\det(M_1+M_2)^N}
\ee
where $\mathcal Z_N^{(2)}$ (the {\it partition function}) is a normalization constant, while  $\alpha'(M_1), \beta'(M_2)$ stand for the product of the densities $\alpha', \beta'$ (the Radon--Nikodym derivatives of the measures $\d \alpha,\d\beta$ with respect to the Lebesgue measure)  over the (positive) eigenvalues of $M_j$.

This probability space is similar to the two--matrix model discussed briefly above for which the coupling between matrices is ${\rm e}^{N \tr M_1M_2}$ \cite{EynardMehta} instead of $\det(M_1+M_2)^{-N}$. The connection with our BOPs (\ref{CBOPs}) is analogous to the connection between ordinary orthogonal polynomials and the Hermitean Random matrix model \cite{MehtaBook}, whose probability space is the set of Hermitean matrices $\mathcal H_N$ equipped with the measure 
$
\d\mu_1(M):= \frac 1{\mathcal Z_N^{(1)}}{\alpha'(M)} \d M. 
$
In particular, we show in \cite{Paper2} how the statistics of the eigenvalues of the two matrices $M_j$ can be described in terms of the biorthogonal polynomials we are introducing in the present work. A prominent role in the description of that statistics  is played by the generalized Christoffel--Darboux identities we develop in Section \ref{Section6}. 

We now summarize the main results of the paper:
\begin{itemize}
\item [-] for an arbitrary totally positive kernel $K(x,y)$ and arbitrary positive measures $\d\alpha,\d \beta$ on $\R_+^2$ we prove that the matrix of bimoments $I_{ab}:= \iint_{\R_+^2} x^a y^b K(x,y)\d \alpha(x) \d\beta( y)$ is totally positive ({\bf Thm. \ref{thm:I}});
\item[-] this implies that there exist, unique, sequences of monic polynomials of degree $n$, $\wt p_n(x), \wt q_n(y)$ biorthogonal to each other as in (\ref{Korth}); we prove that they have {\bf positive and simple} zeroes ({\bf Thm. \ref{thm:alphazeros}});
\item [-] we then specialize to the kernel $K(x,y)= \frac 1{x+y}$; in this case the {\bf zeroes} of $\wt p_n(x)$ ($\wt q_n(y)$) {\bf are interlaced} with the zeroes of the neighboring polynomials
({\bf Thm. \ref{Sturm} });
\item [-] they solve a {\bf four--term} recurrence relation  as specified after \eqref{CBOPs} ({\bf Cor. \ref{fourterm}});
\item [-] they satisfy {\bf Christoffel--Darboux identities} ({\bf Prop. \ref{propCDI}, Cor. \ref{cor:CDIuni}, Thms. \ref{thm:ECD1}, \ref{thm:ECD2}})
\item [-] they solve a {\bf Hermite-Pad\'{e} } approximation problem to a novel type of 
Nikishin systems ({\bf Sec. \ref{sec:AproxPerfD}, Thms. \ref{thm:Padeq}, \ref{thm:Padep}});
\item [-] they can be characterized by a $3\times 3$ {\bf Riemann--Hilbert problems}, ({\bf Props. \ref{RHP1}, \ref{RHP2}}) ;
\end{itemize}
 In the follow-up paper we will explain the relation of the asymptotics of the BOPs introduced in this paper with 
a rigorous asymptotic analysis for continuous (varying) measures $\d\alpha, \d \beta$ using the nonlinear steepest descent method \cite{Paper3}.

\section{Biorthogonal polynomials associated to a totally positive kernel} %
\label{SecPosKer}
 As one can see from the last section 
 %(see for example the definition of $Z$ in \eqref{def:WZ}), 
 the kernel 
$K(x,y)=\frac{1}{x+y}, x,y >0$, which we will refer to as the Cauchy kernel, 
plays a significant, albeit mysterious, role.  We now turn to explaining the 
role of this kernel.  
We recall, following \cite{Karlin}, the definition of the totally positive 
kernel.  
\bd 
\label{posker}
A real function $K(x,y)$ of two variables ranging over linearly ordered 
sets $\mathcal X$ and $\mathcal Y$, respectively, is said to be 
totally positive (TP) if for all 
\be
x_1<x_2<\cdots < x_m, \quad y_1<y_2<\cdots < y_m \quad x_i\in \mathcal X, y_j\in \mathcal Y, 
m \in \N
\ee 
we have 
\be
\det\le[K(x_i,y_j)\ri]_{1\leq i,j\leq m} >0
\ee
\ed 
We will also use a discrete version of the same concept.  
\bd 
A matrix $A:=[a_{ij}],\, i,j=0,1,\cdots n$ is said to be totally positive (TP) if  all its minors are 
strictly positive.  A matrix $A:=[a_{ij}],\, i,j=0,1,\cdots n$ is said to be totally nonnegative (TN) if  all its minors are nonnegative.  A TN matrix $A$ is said to be oscillatory if some positive integer power 
of $A$ is TP.  

\ed 
Since we will be working with matrices of infinite size we introduce a concept of the 
principal truncation.  
\bd A finite $n+1$ by $n+1$ matrix $B:=[b_{i,j}], i,j=0,1, \cdots n$ is said to be the principal truncation of an infinite matrix 
$A:=[a_{ij}],\, i,j=0,1,\cdots \, $ if $b_{i,j}=a_{i,j}, i,j=0,1, \cdots n$.  
In such a case $B$ will be denoted $A[n]$.  
\ed 
Finally, 
\bd An infinite matrix 
$A:=[a_{ij}],\, i,j=0,1,\cdots $  is said to be TP (TN) if $A[n]$ is TP (TN) for every $n=0,1,\cdots$.  
\ed

 \bd Basic Setup \label{def:K}

Let $K(x,y)$ be a {\bf totally positive kernel} on $\R_+\times \R_+$ 
and let $\d\alpha,\d \beta$ be two Stieltjes measures on $R_+$.  We make two simplifying assumptions to avoid degenerate cases: 
\begin{enumerate}
\item $0$ is not an atom of either of the measures (i.e. $\{0\}$ has zero measure).  
\item $\alpha$ and $\beta$ have infinitely many points of increase.  
\end{enumerate}
We furthermore assume: 
\begin{enumerate}
\item [3.] the polynomials are dense in the corresponding Hilbert spaces $H_{\alpha}:=L^2(\R_+,\d \alpha)$, 
 $H_{\beta}:=L^2(\R_+,\d \beta)$, 
 \item [4.]the map 
$ 
 \ds K: H_{\beta}\rightarrow H_{\alpha}$, $\ds Kq(x):=\int K(x,y)q(y)\d \beta(y)
 $
 is bounded, injective and has a dense range in $H_{\alpha}$.  
 \end{enumerate}
 \ed
 Under these assumptions $K$ provides a non-degenerate pairing between $H_{\beta}$ and $H_{\alpha}$: 
 \be \label{def:pairing}
 \langle a | b\rangle= \iint a(x) b(y) K(x,y){\d\alpha\d\beta}, \quad a\in H_{\alpha}, b\in H_{\beta}.  
 \ee 
 \br Assumptions ~3 and ~4 could be weakened, especially the density 
 assumption, but we believe the last two assumptions are the most natural to work with in the Hilbert space set-up of the theory.  \er
 
 Now, let us 
 consider the matrix $ \I$ of generalized bimoments
 \be \label{eq:bimoments}
[\I]_{ij} =  I_{ij}:= \int\int x^i y^j K(x,y) \d \alpha(x) \d\beta(y)\ . 
 \ee

 \bt \label{thm:I}
The semiinfinite matrix $\I$ is TP. 
 \et
 \begin{proof}
  According to a theorem of Fekete, (see Chapter 2, Theorem 3.3 in \cite {Karlin} ), we only need to consider minors of consecutive rows/columns.
 Writing out the determinant, 
 \bew
\Delta_n^{ab}:=  \det [I_{a+i, b+j}]_{0\leq i,j\leq n-1}
 \eew
 we find 
 \begin{align*}
 && \Delta_n^{ab} = \sum_{\sigma\in S_n}\epsilon(\sigma) \iint \prod_{j=1}^n x_j^a y_j^b \prod_{j=1}^n x_j^{\sigma_j-1} y_j^{j-1} K(x_j,y_j) \d^n\alpha(X) \d^n\beta(Y) =\\
 &&
 \iint C(X)^aC(Y)^b \Delta(X) \prod_{j=1}^{n} y_j^{j-1} \prod_{j=1}^{n} K(x_j,y_j) \d^n\alpha \d^n \beta.
 \end{align*}
 Since our intervals are subsets of $\R_+$ we can absorb the powers of $C(X), C(Y)$ into the measures to simplify the notation.  
 Moreover, the function $S(X,Y):= \prod_{j=1}^{n} K(x_j, y_j)$ enjoys the following simple property
 \bew
 S(X, Y_\sigma) = S (X_{\sigma^{-1}} ,Y)\, 
 \eew
 for any $\sigma \in S_n$.  Finally, the product measures $\d^n\alpha = \d ^n\alpha(X), \d^n\beta = \d^n \beta(Y)$ are clearly permutation invariant.
 
 Thus, without any loss of generality, we only need to show that 
 \bew
D_n:=  \iint \Delta(X) \prod_{j=1}^{n} y_j^{j-1} S(X,Y) \d^n\alpha \d^n \beta\, >0, 
 \eew
which is tantamount to showing positivity for $a=b=0$.   
First, we symmetrize $D_n$ with respect to the variables $X$; this produces
 \begin{align*}
 D_n = \frac 1{n!} \sum_{\sigma \in S_n}  \iint \Delta(X_\sigma) \prod_{j=1}^{n} y_j^{j-1} S(X_\sigma,Y) \d^n\alpha \d^n \beta = 
  \frac 1{n!} \sum_{\sigma \in S_n}  \iint \Delta(X)\epsilon(\sigma)  \prod_{j=1}^{n} y_j^{j-1} S(X,Y_{\sigma^{-1}} ) \d^n\alpha \d^n \beta = \cr
  \frac 1{n!} \sum_{\sigma \in S_n}  \iint \Delta(X)\epsilon(\sigma)  \prod_{j=1}^{n} y_{\sigma_j}^{j-1} S(X,Y) \d^n\alpha \d^n \beta 
   = \frac 1{n!}  \iint \Delta(X)  \Delta(Y)  S(X,Y)  \d^n\alpha \d^n \beta. 
 \end{align*}
 Subsequent symmetrization over the $Y$ variables does not change the value of the integral and we obtain (after restoring the definition of $S(X,Y)$)
 \begin{align*}
 D_n = \frac 1{(n!)^2} \sum_{\sigma\in S_n} \epsilon(\sigma)  \iint \Delta(X)  \Delta(Y)  \prod_{j=1}^{n} K(x_j,y_{\sigma_j})  \d^n\alpha \d^n \beta 
 =\\
\frac 1{(n!)^2}   \iint \Delta(X)  \Delta(Y)  \det[ K(x_i,y_j)]_{i,j\leq  n} \d^n\alpha \d^n \beta.   
 \end{align*}
Finally, since $\Delta(X)  \Delta(Y)  \det[ K(x_i,y_j)]_{i,j\leq  n} \d^n\alpha \d^n \beta$ is permutation invariant, it suffices to integrate over the region
$0<x_1<x_2<\cdots <x_n \times 0<y_1<y_2<\cdots <y_n$, and, as a result
\be \label{prinmin}
D_n=\iint_{\substack{0<x_1<x_2<\cdots <x_n \\
0<y_1<y_2<\cdots <y_n}} \Delta(X)  \Delta(Y)  \det[ K(x_i,y_j)]_{i,j\leq  n} \d^n\alpha \d^n \beta. 
\ee
Due to the total positivity of the kernel $K(x,y)$ the integrand is a positive function of all variables and so the integral must be strictly positive.
 \end{proof}
To simplify future computations we define
$
[x] := (1,x,x^2, \dots)^T
$
so that the matrix of generalized bimoments \eqref{eq:bimoments} is simply given by: 
$
\I = \langle [x]|[y]^T\rangle.\label{219}
$
 Now, let $\Lambda$ denote the semi-infinite upper shift matrix. Then we observe that multiplying the measure $\d\alpha(x)$ by $x^i $ or, 
multiplying $ \d \beta(y)$ by $y^j $, is tantamount to multiplying 
 $\I$ on the left by $\Lambda^i$, or 
on the right by $(\Lambda^T )^j$ respectively, which gives us a whole family of bimoment matrices 
 associated with the same $K(x,y)$ but different measures.  Thus we have
 \bc \label{cor:ILambda}
 For any nonnegative integers $i,j$ the matrix of generalized bimoments 
 $ \Lambda^i \I (\Lambda^T)^j$ is TP. 
 \ec

We conclude this section with a few comments about the scope of Theorem \ref{thm:I}.   \br
 Provided that the negative moments are well defined, the theorem then applies to the doubly infinite matrix $I_{i,j}$, $i,j\in \Z$.
 \er
 \br
 If the intervals are $\R$ and $K(x,y) = {\rm e}^{xy}$ then the proof above fails because we cannot re-define the measures by multiplying by powers of the variables, since they become then signed measures, so in general the matrix of bimoments is {\bf not} totally positive. Nevertheless the proof above shows (with $a=b=0$ or $a,b\in 2\Z$) that the matrix of bimoments is positive definite and --in particular-- the biorthogonal polynomials always exist, which is known and proved in \cite{KenNick}.
 \er
 \subsection {Biorthogonal polynomials}
 Due to the total positivity of the matrix of bimoments in our setting, there exist uniquely defined two sequences of monic polynomials 
 \bew
 \widetilde p_n(x) = x^n + \dots\ , \  \ \widetilde q_n(y)= y^n + \dots
 \eew
 such that 
 \bew
 \iint \widetilde p_n(x) \widetilde q_m(y) K(x,y) \d\alpha(x) \d\beta(y) = h_n \delta_{mn}\label{Korth}\ .
 \eew
 Standard considerations (Cramer's Rule) show that  they are provided by the following formul\ae
 \bea
 \widetilde p_n(x) = \frac 1{D_n} \det \le[ \begin{array}{ccc|c}
 I_{00}&\dots &I_{0n-1}& 1\cr
 \vdots &&\vdots&\vdots\cr
 I_{n0}&\dots &I_{nn-1}&x^n
 \end{array}
 \ri]
   \qquad \widetilde q_n(y) = \frac 1{D_n} \det \le[ \begin{array}{ccc}
 I_{00}&\dots &I_{0n}\cr
 \vdots &&\vdots\cr
 I_{n-10}&\dots &I_{n-1n}\cr
 \hline 
 1&\dots&y^n
 \end{array}
 \ri]\\
 h_n = \frac {D_{n+1}}{D_n} > 0,  
 \eea
where $D_j>0$ by equation \eqref{prinmin}.   
For convenience we re-define the sequence in such a way that they are also {\bf normalized} (instead of monic), by  dividing them by the square root of $h_n$;
 \bea \label{def:pq}
 &p_n(x) =\frac{1} {\sqrt{D_nD_{n+1}}} \det \le[ \begin{array}{ccc|c}
 I_{00}&\dots &I_{0n-1}& 1 \cr
 \vdots &&\vdots&\vdots\cr
 I_{n0}&\dots &I_{nn-1}&x^n
 \end{array}
 \ri],\\
 &q_n(y) = \frac{1}{\sqrt{{D_nD_{n+1}}}} \det \le[ \begin{array}{ccc}
 I_{00}&\dots &I_{0n}\cr
 \vdots &&\vdots\cr
 I_{n-10}&\dots &I_{n-1n}\cr
 \hline 
 1&\dots&y^n
 \end{array}
 \ri] \label{BOPs} 
 \eea
 Thus $\langle p_n| q_m\rangle= \delta_{nm}$.

We note also that the BOPs can be obtained by triangular transformations of $[x], [y]$
\be
\p(x) =S_p  [x]\ ,\ \ \ \q(y) = S_q[y]\ ,\qquad [x]=[1,x,x^2,\dots]^t\label{220} 
\ee
where $S_{p,q}$ are (formally) invertible lower triangular matrices 
such that $S_p^{-1}(S_q^{-1})^T=\I$, where, we recall, $\I$ is the 
generalized bimoment matrix.  Moreover, our BOPs satisfy, by construction, 
the recursion relations:
 $$
 x p_i(x) = X_{i,i+1}p_{i+1}(x)+X_{i,i}p_i(x)+\cdots X_{i,0}p_0(x), \qquad 
 y  q_i(y) = Y_{i,i+1}q_{i+1}(y)+Y_{i,i}q_i(y)+\cdots Y_{i,0}q_0(y ), 
 $$
 which will be abbreviated  as 
 \be \label{def:XY}
 x \p(x) = \X\p(x)\ ,\ \ y \q(y) ^T= \q(y)^T\Y^T, 
 \ee
 where $\X$ and $\Y$ are Hessenberg matrices with positive entries on the supradiagonal, 
 and $\p(x)\, \q(y)$ are infinite column vectors $\p(x)^T:=(p_0(x),p_1(x),p_2(x),\dots)^t, \, 
 \q(y)^T:=(q_0(y),q_1(y),q_2(y),\dots)^T$ respectively.   
 
 The biorthogonality can now be written as 
$
\langle \p | \q^T\rangle= Id
$ 
where $Id$ denotes the semi-infinite identity matrix.  
Moreover 
\be \label{eq:XY}
\langle x\p | \q^T\rangle =  \X\ ,\qquad \langle \p|y\q^T\rangle = \Y^T
\ee 
\br The significance of the last two formulas lies in the fact that the operator of multiplication 
is no longer symmetric with respect to the pairing $\langle \bullet|\bullet\rangle$ and as a result the matrices 
$\X$ and $\Y^T$ are distinct.   \er

 \subsection{Simplicity of the zeroes}
In this section we will use the concept of a Chebyshev system of order $n$ and a closely related concept of a Markov sequence.  We refer to \cite{ns} and \cite{gantmacher-krein} for more information.   
The following theorem is a convenient restatement of Lemma 2 in \cite{gantmacher-krein}, p.137.  For easy display we replace determinants with wedge products.  
 \bt \label{thm:CS}
Given a system of continuous functions $\{u_i(x)|i=0\cdots n\}$ let us define the vector field 
\be
\u(x) =\begin{bmatrix} u_0(x), &u_1(x), & \hdots , & u_n(x)  \end{bmatrix}^T, \qquad x\in U. 
\ee 
Then 
$\{u_i(x)|i=0\cdots n\}$ is a Chebyshev system of order $n$ on $U$ 
iff 
the top exterior power
\be 
\u(x_0)\wedge \u(x_1)\wedge \cdots \u(x_n) \ne 0
\ee 
for all $x_0<x_1<\cdots <x_n$ in $U$.  
Furthermore, for $\{u_i(x)|i=0\cdots \}$, if we denote the truncation of $\u(x)$ to the 
first $n+1$ components by $\u_n(x)$, then $\{u_i(x)|i=0\cdots \}$ is a Markov system 
iff
the top exterior power
\be 
\u_n(x_0)\wedge \u_n(x_1)\wedge \cdots \u_n(x_n) \ne 0
\ee 
for all $x_0<x_1<\cdots <x_n$ in $U$ and all $n\in \N$.

\et
 The following well known theorem is now immediate
 
 \bt \label{thm:signchange} 
Suppose $\{u_i(x)|i=0\cdots n\}$ is a Chebyshev system of order $n$ on $U$, 
and suppose we are given $n$ distinct points $x_1, \cdots x_n$ in $U$. 
Then, up to a multiplicative factor, the only 
generalized polynomial $ P(x)=\sum_{i=0}^n a_i u_i(x)$, which vanishes precisely at $x_1, \cdots x_n$ in $U$ 
is given by 
\be 
P(x)=\u(x)\wedge \u(x_1)\wedge \cdots \u(x_n) 
\ee 
\et 
\bt \label{thm:usignchange} 
Denote by $u_i(x)=\int K(x,y)y^i d\beta(y), \, i=0\cdots n$.  
Then $\{u_i(x)|i=0\cdots n\}$ is a Chebyshev system of order $n$ on $\R _+$.  
Moreover, $P(x)$ as defined in Theorem \ref{thm:signchange} changes sign each 
time $x$ passes through any of the zeros $x_j$.  
\et 
\begin{proof}
It is instructive to look at the computation.  Let $x_0<x_1<\cdots x_n$, then using multi-linearity of the exterior product, 
\begin{align*}
&P(x_0)=\u(x_0)\wedge \u(x_1)\wedge \cdots \u(x_n) =\\
&\int 
K(x_0,y_0)K(x_1,y_1)\cdots K(x_n,y_n)[y_0]_n\wedge [y_1]_n\wedge \cdots
\wedge [y_n]_nd\beta(y_0)\cdots d\beta(y_n)= \\
&\frac{1}{n!}\int 
\det[K(x_i, y_j)]_{i,j=0}^n \Delta(Y)d\beta(y_0)\cdots d\beta(y_n)= 
\int_ 
{y_0<y_1<\cdots y_n}\det[K(x_i, y_j)]_{i,j=0}^n 
\Delta(Y) d\beta(y_0)\cdots d\beta(y_n),   
\end{align*} 
where 
$
[y]_n=\begin{bmatrix}y^0, &y^1, &\hdots &y^n \end{bmatrix}^T.  
$
Thus $P(x_0)>0$.  The rest of the proof is the argument about the sign of the 
integrand. To see how sign changes we observe that 
the sign of $P$ depends only on the ordering of 
$x, x_1, x_2, \cdots x_n$, in view of the total positivity of 
the kernel.  In other words, the sign of $P$ is $sgn(\pi)$ where 
$\pi$ is the permutation rearranging $x, x_1, x_2, \cdots x_n$ in an increasing sequence.    
\end{proof}
\bc \label{cor:f}
Let $\{f_i(x):=\int K(x,y)q_i(y) d\beta(y), | i=0\cdots \}$.  
Then $\{f_i(x)|i=0\cdots n\}$ is a Markov sequence  on $\R _+$,
\ec
\begin{proof}
Indeed, Theorem \ref{thm:CS} implies that the group $GL(n+1)$ acts on 
the set of Chebyshev systems of order $n$.  It suffices now to observe 
that $q_j$ are obtained from $\{1,y, \cdots, y^n\}$ by an invertible 
transformation.   
\end{proof}
\br Observe that  $\{f_i(x)|i=0\cdots n\}$ is a Markov sequence regardless of biorthogonality. \er 
Biorthogonality enters however in the main theorem
\bt \label{thm:alphazeros}
The zeroes of $p_n, q_n$ are all simple and  positive. They fall within the convex hull of the support of the measure $\d\alpha$ (for $p_n$'s) and $\d\beta$ (for the $q_n$'s).
\et 
\begin{proof} We give first a proof for $p_n$.  
 The theorem is trivial for $n=0$.  For $1\leq n$ , let us 
suppose $p_n$ has $r<n$ zeros of odd order in  the convex full of $supp(\d \alpha)$.  In full analogy with 
the classical case, $1\leq r$, since 
\bew  
\int p_n(x) f_0(x)d\alpha(x)=\iint p_n(x)K(x,y)d\alpha(x) d\beta(y)
=0
\eew
by biorthogonality, forcing, in view of positivity of $K(x,y)$, $p_n(x)$ 
to change sign in the convex hull of 
$supp(\d\alpha)$.  In the general case, 
denote the zeros by $x_1<x_2<\cdots x_r$.  Using a Chebyshev 
system $f_i(x), i=0, \cdots r$ on $\R_+$ we can construct a unique, up to a 
multiplicative constant, generalized polynomial which vanishes exactly 
at those points, namely 
\be 
R(x)=F(x)\wedge F(x_1)\wedge F(x_2)\wedge \cdots \wedge F(x_r)
\ee
where 
\begin{equation*}
F(x) =\begin{bmatrix} f_0(x) &  f_1(x)& \cdots & f_r(x)  \end{bmatrix}^t, \qquad x\in \R. 
\end{equation*}
It follows then directly from biorthogonality 
that 
\be 
\int p_n(x) F(x)\wedge F(x_1)\wedge F(x_2)\wedge \cdots \wedge F(x_r)d\alpha(x)
=0
\ee
On the other hand, $R(x)$ is 
proportional to $P(x)$ in Theorem \ref{thm:signchange} which, by Theorem \ref{thm:usignchange}, changes sign at each of its zeroes,� so 
the product $p_n(x)R(x)$ is nonzero and of fixed sign over $\R _+\setminus \{x_1,
x_2,\cdots, x_r\}$.  Consequently, the integral is nonzero, since $\alpha$ is assumed to have 
infinitely many points of increase.  Thus, in view of the contradiction, $r\geq n$, hence $r=n$, for $p_n$ is a polynomial of degree $n$.  
The case of $q_n$ follows by 
observing that the adjoint $K^*$ is also a TP kernel and hence it suffices 
to switch $\alpha$ with $\beta$ throughout the argument given above.  \end{proof}

\bl 
In the notation of Corollary \ref{cor:f}
$f_n(x)$ 
has $n$ zeros and $n$ sign changes in the convex hull of $supp(\d \alpha)$.  
\el 
\begin{proof} Clearly, since $\{u_i(x)|i=0\cdots n\}$ is a Chebyshev system of order $n$ on $\R _+$, the number of zeros of $f_n$ cannot be greater 
than $n$.  Again, from 
\bew 
\int f_n(x) p_0(x) \d \alpha(x)=0,
\eew
we conclude that $f_n$ changes sign at least once within the convex hull of $supp(\d\alpha)$.  Let then $x_1<x_2<\cdots x_r$, $1\leq r\leq n$ be all zeros of $f_n$ within the convex hull of $supp(\d\alpha)$ at which $f_n$ changes its sign. Thus, on one hand, 
\bew 
\int \epsilon \ \prod _{i=1}^r (x-x_i) f_n(x)\d\alpha(x) >0, \qquad \epsilon =\pm, 
\eew
while, on the other hand, using biorthogonality we get 
\bew 
\int \epsilon \ \prod _{i=1}^r (x-x_i) f_n(x)\d\alpha(x) =0, \qquad \epsilon =\pm, 
\eew
which shows that $r=n$.  \end{proof}

In view of Theorem \ref{thm:signchange}
the statement about the zeros of $f_n$ has the following 
corollary
\bc {\bf Heine-like representation for $f_n$}
\be 
f_n(x)=C u(x)\wedge u(x_1)\wedge u(x_2) \cdots \wedge u(x_n)
\ee 
where $x_j$ are the zeros of $f_n$  and $C$ is a constant.  
\ec

 \par \vskip 5pt
 
 \section{Cauchy BOPs}
 From now on we restrict our attention to the particular case of the totally positive kernel, 
 namely, the Cauchy kernel
 \be \label{eq:Ckernel}
 K(x,y) = \frac 1{x+y}
 \ee
 whose associated biorthogonal polynomials will be called  Cauchy BOPs .
 Thus, from this point onward, we will be studying the general properties of BOPs for the pairing 
 \be
 \iint p_n(x) q_m(y) \frac {\d\alpha(x) \d\beta(y)} {x+y}= \langle p_n|q_m\rangle\ .
 \ee
 Until further notice, we do not assume anything about the relationship between the two measures $\d \alpha, \d\beta$, other than what is in the basic setup of Definition 
\ref{def:K}.

 \subsection{Rank One Shift Condition}
 It follows immediately from equation \eqref{eq:Ckernel} that 
 \be
 I_{i+1,j} + I_{i,j+1} = \langle x^{i+1}|y^j\rangle +\langle x^i|y^{j+1}\rangle= \int x^i \d\alpha \int y^j \d\beta\ ,
 \ee
 which, with the help of the shift matrix $\Lambda$ and the matrix of 
bimoments $\I$, can be written as: 
 \begin{align*}
& \Lambda \I + \I \Lambda^T=\boldsymbol \alpha \boldsymbol \beta^T,\\
 &\boldsymbol \alpha = (\alpha_0, \alpha_1,\dots)^T\ ,\ \ \alpha_j = \int x^j \d \alpha(x)>0,  \\
 &\boldsymbol\beta  = (\beta_0, \beta_1, \dots)^T\ ,\ \ \beta_j = \int y^j \d\beta(y)>0.  
 \end{align*}
Moreover, by linearity and equation \eqref{eq:XY}, we have
\be \label{eq:XYT}
\X + \Y^T  = {\boldsymbol \pi} {\boldsymbol \eta}^T\ ,\quad 
{\boldsymbol \pi} := \int \p \d\alpha\ ,\ \ {\boldsymbol \eta}:= \int \q \d\beta\ ,\quad
{\p}(x):= (p_0(x),p_1(x),\dots)^t\ ,\ {\q}(y):= (q_0(y),q_1(y),\dots)^t
\ee
which connects the multiplication operators in $H_{\alpha}$ and $H_{\beta}$.  
Before we elaborate on the nature of this connection we 
need to clarify one aspect of equation \eqref{eq:XYT}.  
\br One needs to exercise a great deal of caution using the matrix relation 
given by equation \eqref{eq:XYT}.  Its only rigorous meaning is in 
action on vectors with finitely many nonzero entries or, equivalently, 
this equation holds for all principal truncations. 
\er 

 \bp
\label{etapi}
The vectors ${\boldsymbol \pi}, {\boldsymbol \eta}$ are strictly positive (have nonvanishing positive coefficients).
\ep

%%%%%%%%%%%%%%
\begin{proof}
We prove the assertion only for ${\boldsymbol \pi}$, the one for ${\boldsymbol \eta}$ being obtained by interchanging the roles of $\d\alpha$ and $\d\beta$.

~From the expressions (\ref{BOPs}) for $p_n(x)$ we immediately have
\bea
\pi_n = \sqrt{\frac 1{ D_nD_{n+1}}} \det
\le[ \begin{array}{ccc|c}
I_{00}&\dots &I_{0n-1}& \alpha_0 \cr
\vdots &&\vdots&\vdots\cr
I_{n0}&\dots &I_{nn-1}&\alpha _n
\end{array}
\ri].  
\eea
Since we know that $D_n>0$ for any $ n\geq 0$ we need to prove the positivity of the other determinant. 
Determinants of this type were studied in Lemma 4.10 in \cite{ls-cubicstring}.

We nevertheless give a complete proof of positivity.  First, we
observe that
\bea
\pi_n \sqrt{D_{n+1}D_n} &&=  \sum_{\sigma \in S_{n+1}}\epsilon(\sigma) \int \prod_{j=1}^{n+1} x_j^{\sigma_j -1} \prod_{j=1}^{n} y_j^{j-1}
\frac{\d^{n+1}\alpha\d^n\beta}{\prod_{j=1}^n(x_j + y_j)}= \cr
&& = \int \Delta(X_1^{n+1}) \prod_{j=1}^{n} y_j^{j-1}
\frac{\d^{n+1}\alpha\d^n\beta}{\prod_{j=1}^n(x_j + y_j)}. 
\eea
Here the symbol $X_1^{n+1}$ is to remind that the vector consists of $n+1$ entries (whereas $Y$ consists of $n$ entries) and that the Vandermonde determinant is taken accordingly. Note also that the variable $x_{n+1}$ never appears in the product in the denominator.
Symmetrizing the integral in the $x_j$'s with respect to labels $j=1, \dots, n$ , but leaving $x_{n+1}$ fixed, gives
\bea
\pi_n\sqrt{D_{n+1}D_n} = \frac 1 {n!} \int \Delta(X_1^{n+1}) \Delta(Y) \frac{\d^{n+1}\alpha\d^n\beta}{\prod_{j=1}^n(x_j + y_j)}. 
\eea
Symmetrizing now with respect to the whole set $x_1, \dots, x_{n+1}$ we obtain
\bea
\pi_n\sqrt{D_{n+1}D_n} = \frac 1 {n!(n+1)!} \int \Delta(X_1^{n+1}) \Delta(Y)
\det \le[
\begin{array}{ccc}
K(x_1,y_1) & \dots & K(x_{n+1},y_1)\cr
\vdots &&\vdots \cr
K(x_1,y_{n}) & \dots & K(x_{n+1},y_{n})\cr
1& \dots  & 1
\end{array}\ri]
\d^{n+1}\alpha\d^n\beta
\eea

Moreover, since the integrand is permutation invariant, it suffices to integrate over the region
$0<x_1<x_2<\cdots <x_n<x_{n+1} \times 0<y_1<y_2<\cdots <y_n$, and, as a result
\begin{equation}
\begin{split}
&\pi_n\sqrt{D_{n+1}D_n}=\\
&\iint_{\substack{0<x_1<x_2<\cdots <x_{n+1} 
0<y_1<y_2<\cdots <y_n}} \Delta(X_1^{n+1})  \Delta(Y)\det \le[\begin{array}{ccc}
K(x_1,y_1) & \dots & K(x_{n+1},y_1)\cr
\vdots &&\vdots \cr
K(x_1,y_{n}) & \dots & K(x_{n+1},y_{n})\cr
1& \dots  & 1
\end{array}\ri]
\d^{n+1}\alpha\d^n\beta.  
\end{split}
\end{equation}
We thus need to prove that the determinant containing the Cauchy kernel $\frac 1{x+y}$ is positive  for $0<x_1<x_2<\dots <x_{n+1}$ and $0<y_1<y_2<\dots<y_n$.
It is not difficult to prove that
\be
\det \le[
\begin{array}{ccc}
\frac 1{x_1+y_1} & \dots & \frac 1{x_{n+1}+y_1}\cr
\vdots &&\vdots \cr
\frac 1{x_1+y_{n}} & \dots & \frac 1{x_{n+1}+y_{n}}\cr
1& \dots  & 1
\end{array}\ri]
= \frac{\Delta(X_1^{n+1}) \Delta(Y)}
{\prod_{j=1}^{n+1} \prod_{k=1}^{n} (x_j + y_k)}
\ee
and this function is clearly positive in the above range.\end{proof}

%%%%%%%%%%%%%%%
\par \vskip 5pt

 \subsection{Interlacing properties of the zeroes}
~From (\ref{219}), (\ref{220}) and (\ref {def:XY}) the following factorizations are valid for 
all principal truncations:
$$
\I=S_p^{-1} (S_q^{-1})^T\ ,\quad \X= S_p \Lambda  (S_p)^{-1}\ , \quad \Y= S_q \Lambda S_q^{-1}\ .
$$
Moreover, since $\I$ is TP, the triangular matrices $S_p^{-1} $ and $S_q^{-1} $
are totally nonnegative (TN) \cite{Cryer2} and have the same diagonal entries: the $n$th
diagonal entry being $\sqrt{D_n/D_{n-1}}$.  Furthermore, one can amplify the statement 
about $S_p^{-1} $ and $S_q^{-1} $ using another result of Cryer (\cite{Cryer1}) which implies that 
both triangular matrices are in fact triangular TP matrices (all non-trivial in the sense defined 
in \cite{Cryer1} minors are strictly positive).  This has the immediate consequence 
\begin{lemma}\label{lem:IXIY}
All principal truncations $\X[n], \Y[n]$ are invertible.  
\end{lemma}
\begin{proof}
~From the factorization $\X= S_p \Lambda  (S_p)^{-1} $ we conclude that it suffices to 
prove the claim for $\Lambda S_p^{-1} [n]$ which in matrix form reads: 
\bew 
\begin{bmatrix} (S_p^{-1})_{10}&(S_p^{-1})_{11}&\\
(S_p^{-1})_{20}&(S_p^{-1})_{21}&\hspace{-10pt}(S_p^{-1})_{22}& 
\begin{picture}(0,0)
\put(-50.5,-10){\line(1,-1){20}}
\put(-51,-9.5){\line(1,-1){20}}
\put(-50,-10){\line(1,-1){20}}
\put(0,-10){\hbox{\Huge $0$}}
\end{picture}
 \\
 \\
 \\
\vdots&\vdots&&(S_p^{-1})_{_{n+1, n+1}} \\
(S_p^{-1})_{n+1,0}&(S_p^{-1})_{n+1, 1}&\cdots &(S_p^{-1})_{n+1, n}\\
\end{bmatrix}.  
\eew
However, the determinant of this matrix is strictly positive, because $S_p^{-1}$ is a triangular TP.  

\end{proof}
\br 
This lemma is not automatic, since $\Lambda[n]$ is not invertible.  
\er 

We now state the main theorem of this section.  
\begin{theorem} \label{thm:TN} 
$\X$ and $\Y$ are TN.  
\end{theorem}
\begin{proof}  
We need to prove the theorem for every principal truncation.  Let $n\geq 0$ be fixed.  
We will suppress the dependence on $n$, for example $\X$ in the body of the proof means 
$\X[n]$ etc.  First, we claim that 
$\X$ and $\Y$ admit the L-U factorization:
$ \X = \X_- \X_+,\ \Y= \Y_- \Y_+$, where $A_+$ denotes the upper triangular factor and $A_-$ is the unipotent
lower triangular factor in the Gauss factorization of a matrix $A$. Indeed,
$ \X_+= (\Lambda  S_p^{-1})_+,\ \Y_+= (\Lambda  S_q^{-1})_+$ are upper triangular components of TN matrices
$ \Lambda  S_p^{-1}$ and $\Lambda  S_q^{-1}$ 
 and thus are totally nonnegative invertible bi-diagonal matrices by Lemma \ref{lem:IXIY}.  

~From $\X + \Y^T= \boldsymbol\pi \boldsymbol\eta^T$ we then obtain 
$$ (\Y_+^T)^{-1} \X_- + \Y_- \X_+^{-1} =\left ( (\Y_+^T)^{-1}\boldsymbol\pi\right )
 \left ( \boldsymbol\eta^T  \X_+^{-1} \right ) := \boldsymbol\rho \boldsymbol\mu^T\ .
 $$

We need to show that vectors $ \boldsymbol\rho\ ,\  \boldsymbol\mu$ have positive entries. For this, notice that
\begin{align*}
\boldsymbol\rho&= ((Y_+)^T)^{-1} S_p \boldsymbol\alpha= ( ( (\Lambda S_q^{-1})_+)^T)^{-1} S_p \boldsymbol\alpha\ ,\\
\boldsymbol\mu&=((X_+)^T)^{-1} S_q\boldsymbol\beta=( ( (\Lambda S_p^{-1})_+)^T)^{-1} S_q\boldsymbol\beta.   \end{align*} 

Now, it is easy to check that if the matrix of generalized bimoments $\I$ is replaced 
by $\I\Lambda ^T$ (see Corollary \ref{cor:ILambda} ) then $S_p \rightarrow (( (\Lambda S_q^{-1})_+)^T)^{-1} S_p $, 
while $\boldsymbol\alpha$ is unchanged,  which implies that $\boldsymbol\rho$ is a new  $\boldsymbol\pi $
in the notation of Proposition \ref{etapi}
  and hence positive by the same Proposition.  Likewise, considering 
the matrix of generalized bimoments $\Lambda \I$, for which $\boldsymbol\beta$ is 
unchanged, $S_q \rightarrow (( (\Lambda S_p^{-1})_+)^T)^{-1} S_q$ and $\boldsymbol\mu$ is 
a new $\boldsymbol\eta$ in the notation of Proposition $\ref{etapi}$ implying the claim.  
 
Thus 
$$
\boldsymbol\rho= D_\rho \mathbf 1\ , \boldsymbol\mu = D_\mu \mathbf 1,
$$
where $ D_\rho \ ,  D_\mu $ are diagonal matrices with positive entries and $\mathbf 1$ is a vector  of 1s.

We have
$$ D_\rho^{-1} (\Y_+^T)^{-1} \X_-  D_\mu^{-1} +  D_\rho^{-1}\Y_- \X_+^{-1} D_\mu^{-1}  = \mathbf 1 \ { \mathbf 1 ^T}\ .
 $$
The first  (resp. second) term on the left that we can call $\tilde \X$  (resp.  $\tilde \Y^T$) 
is a lower (resp. upper) triangular matrix with
 positive diagonal entries . The equality above then implies that 
(i) $ \tilde X_{ij} = \tilde Y_{ij} =1$ for all $ i>j$ and (ii) $ \tilde X_{ii} + \tilde Y_{ii} = 1$ for all
$i$. In particular, both $ \tilde X_{ii}$ and $ \tilde Y_{ii} $ are positive numbers strictly less then 1.

This means that $\tilde \X, \tilde \Y$ admits  factorizations
$$ \tilde \X= (Id - \Lambda^T)^{-1} L_X  \ ,\  \tilde \Y= (Id - \Lambda^T)^{-1} L_Y\ ,
$$
where 
$$
L_X= \sum_{i=0}^\infty \tilde X_{ii} E_{ii} + (1- \tilde X_{ii}) E_{i+1\ i}\ , 
L_Y= \sum_{i=0}^\infty \tilde Y_{ii} E_{ii} + (1- \tilde Y_{ii}) E_{i+1\ i}\ .
$$
Since all entries of bi-diagonal matrices $L_X, L_Y$ are positive, these matrices are totally nonnegative
and so are
\be \label{eq:XYFac}
\X= \Y_+^T (Id - \Lambda^T)^{-1} L_X  \X_+\ , \quad \Y= \X_+^T (Id - \Lambda^T)^{-1} L_Y  \Y_+\ .
\ee
\end{proof}

\bc
$\X$ and $\Y$ are oscillatory matrices.  
\ec
\begin{proof}
We give a proof for $\X$. The factorization \eqref{eq:XYFac} we have just obtained shows that $\X$
is the product of an invertible lower-triangular TN matrix $\Y_+^T (Id -\Lambda^T)^{-1}$
and a tri-diagonal
matrix
$J=L_X \X_+$. Note that $L_X$ has all positive values on the main diagonal and
the first
sub-diagonal. 
Entries on the first super-diagonal of $\X_+$ coincide with
corresponding
entries of $\X$ and thus are strictly positive by construction. Moreover, leading
principal minors of $\X$
are strictly positive
(see the proof
of Lemma \ref{lem:IXIY}), which implies that all diagonal entries of $\X_+$ are
strictly positive too.
Thus $J$ is a tri-diagonal matrix with all non-trivial entries strictly
positive.

Since diagonal entries of $\Y_+^T (Id -\Lambda^T)^{-1}$ are strictly positive and all
other entries
are non-negative, 
every zero entry of $\X$ implies that the corresponding entry of $J$ is zero.  
In view of that all entries on the first super- and sub-diagonals of $\X$ must be 
strictly positive, which,
by a fundamental criterion of Gantmacher and Krein (Theorem 10,  II, 
\cite{gantmacher-krein}),
ensures
that $\X$ is oscillatory. 

 \end{proof}

The interlacing properties for the zeros of polynomials $p_n, q_n$, 
as well as other properties of Sturm sequences, 
follow then from Gantmacher-Krein theorems on spectral properties of oscillatory matrices (see II, Theorem 13, in \cite{gantmacher-krein}). 
We summarize the most important properties implied by Gantmacher-Krein theory.  
\bt 
\label{Sturm}
The sequences of BOPs $\{q_n\}$ and $\{p_n\}$ are Sturm sequences.  Moreover, 
\begin{enumerate}
\item their respective zeros are positive and simple,
 \item the roots of adjacent polynomials in the sequences are interlaced,
\item the following alternative representations of the biorthogonal polynomials hold
\begin{align*}
p_n(x)&=\sqrt{\frac{D_n}{D_{n+1}}}\det (x-X[n-1]), \quad 1\leq n, \\
q_n(y)&=\sqrt{\frac{D_n}{D_{n+1}}}\det (y-Y[n-1]), \quad 1\leq n.
\end{align*}
\end{enumerate}
\et
\br The fact that the roots are positive and simple follows indeed from the 
fact that $X$ and $Y$ are oscillatory.  Theorem \eqref{thm:alphazeros}, however, indicates 
that this property is true even for a more general case  when the totally positive kernel $K(x,y)$ 
is not necessarily the Cauchy kernel.  
\er

 \section{Four-term recurrence relations and Christoffel Darboux identities}
 \label{Section6}
 We establish in this section a basic form of recurrence relations and 
an analog of classical Christoffel-Darboux identities satisfied by 
$\{q_n\}$ and $\{p_n\}$.  First, we introduce the following 
 notation for semi-infinite, finite-band matrices.  
 \bd \label{def:matrixsupp}
 Given two integers $a\leq b$ , a semi-infinite matrix $A$ is said to have the support 
 in $[a,b]$ if 
 \be j-i<a \text{ or } j-i>b \text{ imply } A_{ij}=0 \ee
 The set of all matrices with supports in  $[a,b]$ is denoted $M_{[a,b]}$.  
 \ed 
 The content of this section relies heavily on the relation \eqref{eq:XYT} 
which we recall for convenience: 
$$
\X +\Y ^T=\boldsymbol \pi \boldsymbol \eta^T =D_{\pi}\mathbf 1 \mathbf 1 ^T D_{\eta}
$$
 where $D_{\pi}$, $D_{\eta}$ respectively, are diagonal matrices of averages 
of $\p$ and $\q$.   
Since the vector $\mathbf 1$ is a null vector 
of $\Lambda -Id$ we obtain 
\bp \label{thm:XYCR}
$\X$ and $\Y$ satisfy:
\begin{enumerate}
\item $(\Lambda -Id)D_{\pi}^{-1}\X +(\Lambda -Id)D_{\pi}^{-1}\Y^T=0.$
\item
$A:=(\Lambda -Id)D_{\pi}^{-1}\X\in M_{[-1,2]}.$
\item 
$
\X D_{\eta}^{-1}(\Lambda ^T-Id)+\Y ^TD_{\eta}^{-1}(\Lambda ^T-Id)=0.$
\item
 $\wh A:=\X D_{\eta}^{-1}(\Lambda ^T-Id)\in M_{[-2,1]}.$ 
\end{enumerate}
\ep
As an immediate corollary we obtain the factorization property for $X$ and $Y$. 

\bc
Let  $A$, $\wh A$ and 
$$L:=(\Lambda -Id)D_{\pi}^{-1}, \qquad  \wh L:=D_{\eta}^{-1}(\Lambda ^T-Id),$$ 
respectively, denote matrices occurring in Proposition \ref{thm:XYCR}.  Then 
$$ L\X =A, \quad \X\wh L=\wh A, \qquad A\in M_{[-1,2]}, \, \wh A\in M_{[-2,1]}. $$
Likewise, $\Y$ admits a similar factorization:
$$
\Y L^T=B, \qquad (\wh L ^T) \Y=\wh B, 
$$
where $B=-A^T, \wh B=-\wh A ^T$.  
\ec 

Hence, 
\bc
\label{fourterm}
$\p$ and $\q$ satisfy four-term recurrence relations of the form
\begin{align*}
x\le(\frac{p_n(x)}{\pi _n}-\frac{p_{n-1}(x)}{\pi_{n-1}}\ri)=A_{n-1,n+1}p_{n+1}(x)+A_{n-1,n}p_n(x)+A_{n-1,n-1}p_{n-1}(x)+A_{n-1,n-2}p_{n-2}(x), \\
y\le(\frac{q_n(y)}{\eta_n}-\frac{q_{n-1}(y)}{\eta_{n-1}}\ri)=\wh B_{n-1,n+1}q_{n+1}(y)+\wh B_{n-1,n}q_n(y)+\wh B_{n-1,n-1}q_{n-1}(y)+\wh B_{n-1,n-2}q_{n-2}(y),
\end{align*}
for $1\leq n$ with the proviso that $p_{-1}=q_{-1}=0$.  
\ec 
\begin{proof}
We give the proof for $\p(x)$ in matrix form.  Indeed, from  
$$
x\p(x)=\X \p(x),$$ 
it follows that 
$$
xL\p(x)=L \X \p(x),
$$
hence the claim, since $L \in M_{[0,1]}$ 
and $L\X =A\in M_{[-1,2]}$.  \end{proof}
Let us observe that $\wh L$ has a unique formal inverse, represented by a lower 
triangular matrix.  Let us then define 
$$
\wh \p(x)=\wh L ^{-1} \p(x).  
$$

\bt[Christoffel-Darboux Identities for $\q$ and $\p$]\label{thm:CD1}
\begin{equation}\label{eq:CDI2}
(x+y) \sum_{j=0}^{n-1}  q_j(y) p_j (x) = \q^T(y) [\Pi,(y-\Y^T)\wh L]\wh \p(x)
\end{equation}
where $\Pi := \Pi_n$ is the diagonal matrix $diag (1,1,\dots, 1,0,\dots)$ with $n$ ones (the entries are labeled from $0$ to $n-1$).  The explicit form of the commutators is:
\begin{multline}\label{eq:comm2}
[\Pi ,  (y-\Y^T)\wh L] =\wh A_{n-1,n}E_{n-1,n}-(\frac{y}{\eta_n}+\wh A_{n,n-1})E_{n,n-1}
-\\
\wh A_{n,n-2}E_{n,n-2}-\wh A_{n+1,n-1}E_{n+1,n-1}, 
\end{multline}

where $A_{i,j}$, $\wh A_{i,j}$ respectively, denote the $(i,j)$th entries  of $A$, $\wh A$, occurring in Proposition \ref{thm:XYCR}.  

\et
\begin{proof}
We give the proof of equation \eqref{eq:CDI2}.  Since $(y-\Y)\q=0$ it suffices
to prove that the left hand side equals $\q ^T\Pi(y-Y^T)\wh L\wh \p(x)$. 
From the definition of $\wh \p$ and equation \eqref{def:XY} we obtain
$$
(x+y)\q^T(y)\Pi \p(x)= \q^T(y)\Pi y \wh L \wh \p(x)+
\q^T(y)\Pi \X\p(x)=\q^T(y)\Pi y \wh L \wh \p(x)+
\q^T(y)\Pi \X\wh L \wh \p(x), $$
which, after switching $\X\wh L$ with $-\Y^T \wh L$ in view of Proposition \ref{thm:XYCR}, gives equation \eqref{eq:CDI2}.  To get the commutator equation \eqref{eq:comm2} 
one needs to perform an elementary computation using the definition of $\wh A$.  
\end{proof}

We establish now basic properties of $\wh \p$ and its biorthogonal partner 
$\wh \q$ defined below.

 \bp
 \label{hattedBOPs}
 The sequences of polynomials
 \be
 \wh \p = \wh L^{-1} \p \ , \ \ \ \wh \q ^T= \q  ^T \wh L
 \ee
 are characterized by the following properties 
 \begin{enumerate}
 \item $\deg \wh q_n = n+1$, $\deg \wh p_n = n$;
 \item $\ds \int \wh q_n \d \beta = 0$;
\item $\ds \iint \wh p_n(x) \wh q_m(y) \frac {\d\alpha \d\beta}{x+y} = \delta_{mn}$ ;
\item $\wh q _n(y) =\frac{1}{\eta_{n+1}}\sqrt{\frac{D_{n+1}}{D_{n+2}}}y^{n+1}+
\mathcal O (y^n);$

\end{enumerate}
 In addition
 \begin{description}
 \item a. $\wh \q$ and $\wh \p$ satisfy the intertwining relations with $\q$ and $\p$ 
 \bea
&& y \wh\q^T=  -   \q^T\wh A , \cr
 && x \p =  \wh A\wh \p; \label{recrelhat}
 \eea
\item b. $\wh \q$ and $\wh \p$ admit the determinantal representations:
 \bea
\wh q_n(y) &\&= \frac{1}{\eta_{n}\eta_{n+1}\sqrt{D_nD_{n+2}}} \det \le[
 \begin{array}{cccc}
 I_{00} & \dots &&I_{0n+1}\\
 \vdots& && \vdots\\
 I_{n-1\,0}& \dots&&I_{n-1\,n+1}\\
 \beta_0& \dots & &\beta_{n+1}\\
 1 & \dots &&y^{n+1} 
 \end{array}
 \ri]\\
\wh p_n(x) &\&   =\frac {1}{D_{n+1}} \det \le[
 \begin{array}{cccc}
 I_{00} & \dots &I_{0\,n}&1\\ 
 \vdots& && \vdots\\
 I_{n-1\,0}& \dots&I_{n-1\,n}&x^{n-1}\\
 I_{n 0}& \dots & I_{n\,n}&x^n\\
 \beta_0& \dots &\beta_{n}& 0 
 \end{array}
 \ri]\label{detwhp}
 \eea
 \item c. $\ds \beta_0 \iint \wh p_n(x) y^j  \frac {\d\alpha \d\beta}{x+y}  = \beta_j \iint \wh p_n(x) \frac {\d\alpha \d\beta}{x+y}$, $ j\leq n$.

\end{description}
 \ep

\begin{proof}

Assertions (1), (2) and (4) follow directly from the shape of the matrix $\wh L$. Assertion (3) follows from $\langle \p,\q^t\rangle=\1$ by multiplying it by $\wh L$ on the right and by $\wh L^{-1}$ on the left.
Assertion (c) follows from assertions (1), (2) and (3); indeed from (2) and (3), it follows that  the polynomial $\wh p_n$ is biorthogonal to all polynomials of degree $\leq n$ with zero $\d\beta$--average and $\{\beta_0 y^j - \beta_j: 
0\leq j\leq n\}$ 
is a basis for such polynomials.

The intertwining relations follow from the definitions of the matrices $\wh L, \wh A$ and of the polynomials $\wh \p, \wh \q$.

The determinantal expression for $\wh q_n$ follows by inspection since the proposed expression has the defining properties (1) and (2)  and is biorthogonal to all powers $1,x,\dots, x^{n-1}$. The normalization is found by comparing the leading coefficients of $\wh q_n = \frac{1}{\eta_{n+1}} q_{n+1}+ \mathcal O(y^n)$.
The determinantal expression for $\wh p_n(x)$ follows again by inspection; indeed if $F(x)$ is the determinant in (\ref{detwhp}) then 
\be
\langle F(x)|y^j \rangle=  \det \le[
 \begin{array}{cccc}
 I_{00} & \dots &I_{0\,n}&I_{0j}\\ 
 \vdots& && \vdots\\
 I_{n-1\,0}& \dots&I_{n-1\,n}&I_{n-1\, j}\\
 I_{n 0}& \dots & I_{n\,n}&I_{n\,j}\\
 \beta_0& \dots &\beta_{n}& 0 
 \end{array}
 \ri] = -\beta_j D_{n+1} = \frac{\beta_j}{\beta_0} \langle F(x)|1\rangle.  
\ee
where the determinants are computed by expansion along the last row.
The proportionality constant is again found by comparison.  
\end{proof}

\par\vskip 5pt
One easily establishes a counterpart to Theorem \ref{thm:CD1} valid 
for $\wh \q$ and $\wh \p$.  

\bp[Christoffel--Darboux identities for $\wh \q$ and $\wh \p$ ]
\label{propCDI}
We have
\be \label{CDI1}
(x+y) \sum_{j=0}^{n-1} \wh q_j(y) \wh p_j (x) = \q^T(y)  [(x-X)\wh L, \Pi]\wh  \p(x)= 
\q^T(y) [\Pi,(-x -Y^T)\wh L]\wh \p(x).  
\ee
\ep

\br Observe that the commutators occurring in both theorems have 
identical structure; they only differ in the variable $y$ in Theorem 
\ref{thm:CD1} being now replaced by $-x$.  
We will denote by $\mathbb A(x) $ the commutator $[\Pi, (-x-\Y^T)\wh L]$ and by $\mathbb A_n(x)$ its nontrivial $3\times 3$ block.  Thus the nontrivial block in Proposition \ref{propCDI} reads:
\be
\mathbb A_n(x) = \le[
\begin{array}{cc|c}
0&0&\wh A _{n-1,n}\\
\hline
-\wh A _{n,n-2}& \frac{x}{\eta_{n}} -\wh A _{n, n-1}& 0\\
0&-\wh A_{n+1,n-1}&0
\end{array}\ri]
\ee 
while the block appearing in Theorem \ref{thm:CD1} is simply $\mathbb A_n(-y)$.  
\par \vskip 5pt
\er
 With this notation in place we can present  the Christoffel-Darboux 
identities in a unified way. 

\bc[Christoffel--Darboux identities for $\q, \p$, and $\wh \q,\wh \p$ ] \label{cor:CDIuni}
The biorthogonal polynomials $\q, \p$, and $\wh \q,\wh \p$ satisfy 
\bea
(x+y) \sum_{j=0}^{n-1} q_j(y) p_j (x) = \q^T(y) \mathbb A(-y)\wh \p(x),\\
(x+y) \sum_{j=0}^{n-1} \wh q_j(y) \wh p_j (x) = \q^T(y) \mathbb A(x)\wh \p(x).
 \eea
\ec
\section{ Approximation problems and perfect duality} \label{sec:AproxPerfD}
We will associate a chain of Markov functions associated with measures 
$\d \alpha$ and $\d \beta$ by taking the 
Stieltjes' transforms of the corresponding measures as well as their reflected, with respect to the origin, images.  
%These are abstract analogs of Weyl functions discussed in earlier sections (see Definition \ref{def:WZ}).  
\bd
Define 
\begin{align}
&W_{\beta}(z)=\int \frac{1}{z-y} \d \beta(y),  &W_{\alpha^*}(z)&=\int \frac{1}{z+x}\d \alpha(x), \cr
&W_{\alpha^*\beta}(z)=-\iint \frac{1}{(z+x)(x+y)}\d \alpha(x) \d \beta(y),
&W_{\beta \alpha^*}(z)&=\iint \frac{1}{(z-y)(y+x)}\d \alpha(x) \d \beta(y).   
\end{align}
\ed
We recall now an important notion of a Nikishin system
associated with two measures (see \cite{ns}, p. 142, called there a 
MT system of order $2$).  

\bd Given two measures $\d \mu _1$ and $\d\mu_2$ with disjoint 
supports $\Delta _1 $, $\Delta _2$ respectively, a Nikishin system of order $2$ is a pair of 
functions 
$$
f_1(z)=\int _{\Delta_1} \frac{\d \mu_1(x_1)}{z-x_1}, \qquad 
f_2(z)=\int_{\Delta_1}\frac{\d \mu_1(x_1)}{z-x_1}\int _{\Delta_2} \frac{\d \mu_2(x_2)}{x_1-x_2}.  
$$
\ed
\br The definition of a Nikishin system depends on the order in 
which one "folds" measures.  If one starts from $\d \mu_2$ , rather than 
$\d \mu _1$ one 
obtains a priory a different system.  As we show below the relation between 
these two Nikishin systems is in fact of central importance to the 
theory we are developing.  
\er 
The following elementary observation provides the proper framework for our 
discussion.  
\bl Let $\d \alpha^*$ denote the measure obtained from $\d \alpha$ by 
reflecting the support of $\d \alpha$ with respect to the origin.  Then 
$W_{\beta}, W_{\beta \alpha^*}$ and $W_{\alpha^*}, W_{\alpha^*\beta}$ 
are Nikishin systems associated with measures $\d \beta $ and $\d \alpha^*$ with no predetermined ordering of measures. 
\el  
The relation between these two Nikishin systems can now be readily 
obtained. 
\bl \label{lem:Plucker}
\be \label{eq:Plucker} 
W_{\beta}(z)W_{\alpha^*}(z)=W_{\beta \alpha^*}(z)+W_{\alpha^*\beta}(z). 
\ee
\el
\begin{proof}
Elementary computation gives: 
$$
W_{\beta}(z)W_{\alpha^*}(z)=\iint \frac{1}{(z-y)(z+x)}\d \alpha(x) \d \beta(y)=
\iint \frac{1}{(x+y)}\le[\frac{1}{z-y}-\frac{1}{z+x}\ri ]\d \alpha(x) \d \beta(y), 
$$
which implies the claim. 
\end{proof}
\br Equation \eqref{eq:Plucker} was introduced in \cite{ls-cubicstring} for the DP peakons (see Lemma 4.7 there). 
%This equation represents a generalization of the formula \eqref{eq:WZconstraint} of the present paper. 
Observe that this formula is valid for any Nikishin system of order $2$.   
\er 
We formulate now the main approximation problem, modeled after that of 
\cite{ls-cubicstring} %(see also Definition \ref{def:pade})
\bd Let $n\geq 1$.  Given two Nikishin systems $W_{\beta}, W_{\beta \alpha^*}$ and $W_{\alpha^*}, W_{\alpha^*\beta}$ we seek polynomials $Q(z), degQ=n$, $P_{\beta}(z), 
deg P_{\beta}=n-1$ and $P_{\beta \alpha^*}(z), deg P_{\beta \alpha^*}=n-1$,  
which satisfy Pad$\acute{e}$-like approximation conditions as $z\rightarrow \infty, \, z\in \C_{\pm} $: 
\begin{subequations}\label{eq:PadeA}
\begin{align}
Q(z)W_{\beta}(z)-P_{\beta}(z)=\mathcal O\le(\frac{1}{z}\ri), \\
Q(z)W_{\beta\alpha^*}(z)-P_{\beta\alpha^*}(z)=\mathcal O\le(\frac{1}{z}\ri), \\
Q(z) W_{\alpha^*\beta}(z)-P_{\beta}(z)W_{\alpha^*}(z)+P_{\beta\alpha^*}(z)=\mathcal O\le(\frac{1}{z^{n+1}}\ri)
\end{align}
\end{subequations}

\ed
\br In the case that both measures have compact support we can 
remove the condition that $z\in \C _{\pm}$ since 
all the functions involved are then holomorphic around $z=\infty$.  
\er 
\br In the terminology used for example in \cite{vanassche} the triplets of 
polynomials $Q, P_{\beta},P_{\beta \alpha^*}$ provide a Hermite-Pad\'{e}
approximation of type $I$ to the Nikishin system $W_{\beta}, W_{\beta \alpha^*}$ and, simultaneously, a Hermite-Pad\'{e}
approximation of type $II$ to the Nikishin system $W_{\alpha^*}, W_{\alpha^*\beta}$. \er
 
\bd 
We call the left hand sides of approximation problems \eqref{eq:PadeA} 
$R_{\beta}, R_{\beta \alpha^*}$ and $R_{\alpha^*\beta}$ respectively, referring to them as remainders.  
\ed 

The relation of the approximation problem \eqref{eq:PadeA} to the theory of biorthogonal 
polynomials $\q$ and $\p$ is the subject of the next theorem.  
\bt \label{thm:Padeq} 
Let $q_n(y)$ be defined as in \eqref{BOPs}, and let us set 
$Q(z)=q_n(z)$ Then $Q(z)$ is the unique, up to a multiplicative 
constant, solution of the approximation problem \eqref{eq:PadeA}.  
Moreover, $P_{\beta},P_{\beta \alpha^*}$ and all the remainders 
$R_{\beta}, R_{\beta \alpha^*}$ and $R_{\alpha^*\beta}$ are uniquely determined from $Q$ with the help of the formulas: 
\begin{subequations}
\begin{align}
P_{\beta}(z)&=\int \frac{Q(z)-Q(y)}{z-y}\d \beta(y),\qquad 
P_{\beta \alpha^*}(z)=\iint\frac{Q(z)-Q(y)}{(z-y)(x+y)}\d \alpha(x) \d \beta(y) ,  \\
R_{\beta}(z)&=\int \frac{Q(y)}{z-y}\d \beta(y),  \qquad \qquad 
R_{\beta \alpha^*}(z)=\iint\frac{Q(y)}{(z-y)(x+y)}\d \alpha(x) \d \beta(y) , \\
R_{\alpha^*\beta}(z)&=-\iint\frac{Q(y)}{(z+x)(x+y)}\d \alpha(x) \d \beta(y)
=\int\frac{R_{\beta}(x)}{z-x}      \d \alpha^*(x).  
\end{align}
\end{subequations}

\et 
\begin{proof}
We start with the first approximation problem involving $Q(z)W_{\beta}(z)$.  
Writing explicitly its first term we get: 
$$
\int \frac{Q(z)}{z-y}\d\beta(y)=\int \frac{Q(z)-Q(y)}{z-y}\d\beta(y)+\int \frac{Q(y)}{z-y}\d\beta(y).$$

 Since $ \int \frac{Q(z)-Q(y)}{z-y}\d\beta(y)$ is a polynomial in $z$ of degree $n-1$, while $\int \frac{Q(y)}{z-y}\d\beta(y)=
\mathcal O (\frac{1}{z})$, we get the first and the third formulas.  The second and fourth formulas are obtained in an analogous way from the second approximation problem.  
Furthermore, to get the last formula we compute $P_{\beta}$ and $P_{\beta \alpha^*}$ from the first two approximation problems and substitute into the third 
approximation problem, using on the way Lemma \ref{lem:Plucker}, to obtain: 
$$
R_\beta W_{\alpha^*}-R_{\beta \alpha^*}=R_{\alpha^* \beta}.  
$$
Substituting explicit formulas for $R_{\beta}$ and $R_{\beta \alpha^*}$ gives 
the final formula. To see that $Q(z)$ is proportional to $q_n(z)$ we 
rewrite  $-R_{\alpha^*\beta}$ as: 
\begin{align*}
&\iint \frac{Q(y)}{(z+x)(x+y)}\d \alpha(x)\d \beta(y)=\iint\frac {Q(y)}{(x+y)}
\left[
\frac{1}{z+x}-\frac{1-(-(\frac{x}{z}))^n}{z+x}\right ]\d \alpha(x)\d \beta(y)+\\
&\iint\sum_{j=0}^{n-1}\frac{(-x)^j}{z^{j+1}}\frac{Q(y)}{(x+y)(z+x)}\d \alpha \d \beta =
\iint\frac {Q(y)}{(x+y)}\left[
\frac{(\frac{-x}{z})^n}{z+x}\right]\d \alpha(x)\d \beta(y)+
\iint\sum_{j=0}^{n-1}\frac{(-x)^j}{z^{j+1}}\frac{Q(y)}{(x+y)(z+x)}\d \alpha \d \beta
\end{align*}

To finish the argument we observe that the first term is already $\mathcal O(
\frac{1}{z^{n+1}})$, hence the second term must vanish.  This gives: 
$$
\iint \frac{x^j Q(y)}{x+y}\d\alpha(x) \d \beta(y)=0, \qquad 0\leq j\leq n-1,
$$
which characterizes uniquely (up to a multiplicative constant) 
the polynomial $q_n$.  
\end{proof}
\br In the body of the proof we used an equivalent form 
of the third approximation condition, namely
\be \label{eq:3rdPadeA}
R_\beta W_{\alpha^*}(z)-R_{\beta \alpha^*}(z)=R_{\alpha^* \beta}(z)=\mathcal O(\frac{1}{z^{n+1}}). 
\ee
\er

By symmetry, we can consider the Nikishin systems associated with measures 
$\alpha$ and $\beta^*$ with the corresponding Markov functions 
$W_{\alpha}, W_{\alpha \beta^*}$ and $W_{\beta^*}, W_{\beta^*\alpha}$. 
We then have an obvious interpretation of the polynomials $p_n$.  
\bt \label{thm:Padep}
Let $p_n(x)$ be defined as in \eqref{BOPs}, and let us set 
$Q(z)=p_n(z)$. Then $Q(z)$ is the unique, up to a multiplicative 
constant, solution of the approximation problem for $z \rightarrow 
\infty, z\in \C_{\pm}$:
\begin{subequations}\label{eq:PadeB}
\begin{align}
Q(z)W_{\alpha}(z)-P_{\alpha}(z)=\mathcal O\le(\frac{1}{z}\ri), \\
Q(z)W_{\alpha\beta^*}(z)-P_{\alpha\beta^*}(z)=\mathcal O\le(\frac{1}{z}\ri), \\
Q(z)W_{\beta^*\alpha}(z)-P_{\alpha}(z) W_{\beta^*}(z) +P_{\alpha\beta^*}(z)=\mathcal O\le(\frac{1}{z^{n+1}}\ri), 
\end{align}
\end{subequations}

where $P_{\alpha}, P_{\alpha \beta^*}$ are given by formulas of Theorem 
\ref{thm:Padeq} after switching $\alpha$ with $\beta$.  
\et 

Clearly, one does not need to go to four different types of 
Nikishin systems in order to characterize $q_n$ and $p_n$.  
The following corollary is an alternative characterization of 
biorthogonal polynomials which uses only  the first pair of Nikishin systems. 
\bc \label{cor:Padeqp}
Consider the Nikishin systems $W_{\beta}, W_{\beta \alpha^*}$ and $W_{\alpha^*}, W_{\alpha^*\beta}$.  Then the pair of biorthogonal polynomials $\{q_n, p_n\}$ solves: 
\begin{enumerate}
\item $Q(z)=q_n(z) $ solves Hermite-Pad\'{e} approximations given by equations 
\eqref{eq:PadeA}, 
\begin{subequations}
\begin{align*}
Q(z)W_{\beta}(z)-P_{\beta}(z)=\mathcal O\le(\frac{1}{z}\ri), \\
Q(z)W_{\beta\alpha^*}(z)-P_{\beta\alpha^*}(z)=\mathcal O\le(\frac{1}{z}\ri), \\
Q(z) W_{\alpha^*\beta}(z)-P_{\beta}(z)W_{\alpha^*}(z)+P_{\beta\alpha^*}(z)=\mathcal O\le(\frac{1}{z^{n+1}}\ri)
\end{align*}
\end{subequations}

\item $Q(z)=p_n(-z)$ solves switched (Type I with Type II) Hermite-Pad\'{e} approximations 
\begin{subequations}\label{eq:PadeC}
\begin{align}
Q(z)W_{\alpha^*}(z)-P_{\alpha^*}(z)=\mathcal O\le(\frac{1}{z}\ri), \\
Q(z)W_{\alpha^*\beta }(z)-P_{\alpha^*\beta}(z)=\mathcal O\le(\frac{1}{z}\ri), \\
Q(z) W_{\beta \alpha^*}(z)-P_{\alpha^*}(z)W_{\beta}(z)+P_{\alpha^*\beta}(z)=\mathcal O\le(\frac{1}{z^{n+1}}\ri)
\end{align}
\end{subequations}
\end{enumerate}
\ec 
We finish this section with a few results needed for the Riemann-Hilbert problem approach to biorthogonal polynomials $\{q_n, p_n\}$ which will be presented in the next section.  

\bd
 \label{defauxwave}
 We define the auxiliary vectors
in addition to the main polynomial vectors $\q_{_0}(w):= \q(w)$ and $\p_{_0}(z) := \p(z)$, as 
 \bea
&&  \q_{_1}(w) :=\int \q(y)\frac {\d\beta(y)}{w-y}, \qquad \q_{_2}(w):=\int \frac{\q_1(x)}{w-x}\d \alpha^*(x), \\
&&  \p_{_1} (z) :=\int \frac {\p(x)\d\alpha(x)}{z-x}, \qquad \p_{_2}(z) := \int \frac{\p_1(y)}{z-y}\d \beta^*(y).  
\eea
Moreover,
\bea
&&\wh \p_{_1}(z) := \wh L^{-1}  \le(\p_{_1}(z)+\frac{1}{\beta_0}\langle \p|1\rangle\ri)=\wh L^{-1} \p_{_1}(z) - \1,\\
&& \wh \p_2(z):=\int \frac{\wh \p_1(y)}{z-y}\d \beta^*(y). 
\eea
Here $\1$ is the vector of ones.\footnote{The formula $\beta_0^{-1} <\wh \p_n,1> = -1$ follows directly from the determinantal expression in Proposition  \ref{hattedBOPs}}. 
\ed

\br 
 Note that the definition above  unifies the approximants and their 
respective remainders (see Theorem \ref{thm:Padeq}), thus, for example, $\q_{_1}(w) =\mathbf R_{\beta}(w), 
\q_{_2}(w)=\mathbf R_{\alpha^*\beta}(w)$ etc.  The definition of ``hatted'' quantities 
is justified below.  
\er 
\bt [Extended Christoffel-Darboux Identities]\label{thm:ECD1}
Let ${a,b=0,\dots 2}$.  Then  
\be
(w+z) \q_{_a}^T(w) \Pi  \p_{_b}(z) = \q_{_a}^T(w) \mathbb A(-w) \wh \p_{_b}(z)-\mathbb F(w,z)_{ab}
\ee
where 
\be
\mathbb F(w,z)= \begin{bmatrix}
0&0&1\\
0& 1&W_{\beta^*}(z)+W_{\beta}(w) \\
1&W_{\alpha}(z)+W_{\alpha^*}(w)&W_{\alpha^*}(w)W_{\beta^*}(z)+W_{\alpha^*\beta}(w)+W_{\beta^*\alpha}(z)  
\end{bmatrix}.  
\ee 

\et 
\begin{proof}
The proof goes by repeated applications of the Christoffel-Darboux Identities 
given by Theorem \ref{thm:CD1} and Pad\'{e} approximation conditions 
\ref{eq:PadeA}.  The details have been relegated to Appendix \ref{sec:app-CD}.  
\end{proof} 

We point out that if we set $w=-z$ in the CDI's contained in Theorem 
\ref{thm:ECD1}, the left hand side vanishes identically and the RHS contains terms of the form $\q_{a}(-z) \mathbb A(z) \wh\p_{_b}(z) $ minus $\mathbf F_{ab}(-z,z)$. The main observation is that the second term is {\bf constant}, independent of both $z$ and $n$, and hence one ends up with the 
{\bf perfect pairing} (see \cite{Bertosemiclass})
between the auxiliary vectors. 
For the reader's convenience we recall the definition of $\mathbb A(z)$ to 
emphasize the implicit dependence on the index $n$ hidden in the projection 
$\Pi$.

\bt (Perfect Duality) 

Let 
$$
\mathbb J=\begin{bmatrix}0&0&1\\0&1&0\\1&0&0 \end{bmatrix}.
$$
Then  
$$
\q_{_a}^T(-z) \mathbb A(z) \wh \p_{_b}(z)=\mathbb J_{ab}, \qquad 
\text{ where }
\mathbb A(z)=[(z-\X)\wh L,\Pi]. 
$$
\et
\begin{proof} 
The only nontrivial entry to check is $(2,2)$.  In this case, after one 
substitutes $w=-z$ into 
$W_{\alpha^*}(w)W_{\beta^*}(z)+W_{\alpha^*\beta}(w)+W_{\beta^*\alpha}(z)$, one 
obtains the identity of Lemma \ref{lem:Plucker}. 
\end{proof}
There also exists an analog of the extended Christoffel-Darboux identities 
 of Theorem \ref{thm:ECD1} for the ``hatted'' quantities.  

We first define: 
\bd 
For $a=0,1,2$,
\be
\wh \q_{_a}^T  := \q_{_a}^T \wh L. 
\ee
\ed
The following identities follow directly from the respective 
definitions.  
\bl \label{lem:wqzp} 
\begin{align*}
&w \wh \q_a^T(w)=\begin{cases} \q_a^T(w)\Y^T \wh L, \quad &a=0,1\\
                    \q_2^T(w)\Y^T \wh L -\langle 1|\wh\q_0^T\rangle, \quad &a=2. \end{cases}\\
&(z-\X)\wh L \wh \p_b(z)=\begin{cases}0, \quad &b=0,\\
                             \frac{\langle \p_0|z+y\rangle}{\beta_0}, \quad &b=1,\\
-\langle \p_0|1\rangle +\frac{\langle \p_0|z+y\rangle W_{\beta^*}(z)}{\beta_0}, \quad &b=2.  
\end{cases}
\end{align*}
\el

\bt [Extended Christoffel-Darboux Identities for $\wh \q_a, \wh \p_b$]\label{thm:ECD2}
Let $a,b=0,\dots 2$. Then  
\be
(w+z) \wh \q_{a}^T(w) \Pi \wh \p_{b}(z) = 
\q_{a}^T(w)\mathbb A(z)\wh \p_{b}(z)-\wh {\mathbb F}(w,z)_{ab}
\ee
where 
\be
\wh {\mathbb F}(w,z)={\mathbb F}(w,z) -\frac{w+z}{\beta_0}\begin{bmatrix}
0&1&W_{\beta^*}(z)\\
0& W_{\beta}(z)&W_{\beta}(w)W_{\beta^*}(z)\\
1&W_{\alpha^*\beta^*}(w)&W_{\alpha^*\beta^*}(w)W_{\beta^*}(z)
\end{bmatrix}.  
\ee 

\et 
\begin{proof}
We give an outline of the proof.  
For $a=0,1$, in view of Lemma \ref{lem:wqzp}
$$
(w+z) \wh \q_{a}^T(w) \Pi \wh \p_{b}(z)=\q_{a}^T(w)\mathbb A(z)\wh \p_{b}(z)
+\q_{a}^T(w)\Pi (z-\X)\wh L \wh \p_{b}(z).  
$$
The second term equals, again by Lemma \ref{lem:wqzp}, 
$$
\q_{a}^T(w)\Pi\begin{cases}0, \quad &b=0,\\
                             \frac{\langle \p_0|z+y\rangle}{\beta_0}, \quad &b=1,\\
-\langle \p_0|1\rangle+\frac{\langle \p_0|z+y\rangle W_{\beta^*}(z)}{\beta_0}, \quad &b=2.  
\end{cases}
$$  
Now, one goes case by case, using biorthogonality of $\q_0^T$ and $\p_0$, 
and the definition of $\q_1^T(w)$.  After a few elementary steps one arrives 
at the claimed result.  
The computation for $a=2$ is only slightly more involved.  From Lemma 
\ref{lem:wqzp} we obtain: 
$$
(w+z) \wh \q_{2}^T(w) \Pi \wh \p_{b}(z) = 
\q_{2}^T(w)\mathbb A(z)\wh \p_{b}(z)-\langle 1|\wh \q_0\rangle \Pi \wh \p_b(z)+
\q_2^T(w)\Pi (z-\X)\wh L \wh \p_b(z).  
$$
In view of biorthogonality of $\wh \q_0^T$ and $\wh \p$, after some 
intermediate computations, one obtains: 
\begin{align*}
\langle 1|\wh \q_0\rangle \Pi \wh \p_b(z)=\begin{cases} 1, \quad &b=0\\
                                           W_{\alpha}(z)+\frac{\langle 1|1\rangle}{\beta_0}, 
\quad &b=1,\\
W_{\beta^*\alpha}(z)+\frac{\langle 1|1\rangle }{\beta_0}W_{\beta^*}(z), \quad &b=2.  
\end{cases} 
\end{align*}                                           
Likewise, 
\begin{align*}
\q_2^T(w)\Pi (z-\X)\wh L \wh \p_b(z)=\begin{cases} 0, \quad &b=0\\
                                     \frac{w+z}{\beta_0}W_{\alpha^*\beta}(w)    - W_{\alpha^*}(w)+\frac{\langle 1|1\rangle }{\beta_0}, 
\quad &b=1,\\\frac{w+z}{\beta_0}W_{\beta^*}(z)W_{\alpha^*}(w)-W_{\alpha^*\beta}(w)-W_{\beta^*}(z)W_{\alpha^*}(w)+
\frac{\langle 1|1\rangle }{\beta_0}W_{\beta^*}(z), \quad &b=2,   
\end{cases} 
\end{align*}  
and the claim follows.  

\end{proof}

\section{Riemann--Hilbert problems}
In this section we set up two 
Riemann--Hilbert problems characterizing the  Cauchy BOPs that enter the Christoffel--Darboux identities of the previous section. This is done in anticipation of possible applications to the study of universality for the corresponding two--matrix model.  Moreover, since the Christoffel--Darboux kernels contain also the  hatted polynomials, it is useful to 
formulate the Riemann--Hilbert problems for those polynomials as well.  

We will also make the {\bf assumption} (confined to this section) that the measures $\d\alpha,\d\beta$ are {\it absolutely continuous with respect to Lebesgue's measure} on the respective axes.  Thus one can write $
\frac{\d \alpha}{\d x} = {\rm e}^{-\frac{U(x)}\hbar}\ ,
\frac{\d \beta}{\d y} = {\rm e}^{-\frac{V(y)}\hbar},  
$
for the respective (positive!) densities on the respective supports: the signs in the exponents are conventional so as to have (in the case of an unbounded support) the {\it potentials} $U,V$ bounded from below. The constant $\hbar$ is only for convenience when studying the asymptotics of biorthogonal polynomials 
for large degrees (small $\hbar$).   

Since the Christoffel--Darboux identities involve the expressions $\q_{_a}\mathbb A \wh \p_{_b}$, we are naturally led to characterize the sequences $\q$ and $\wh \p$.  However, the other sequences can be characterized in a similar manner by swapping the r\^oles of the relevant measures and symbols.

\subsection{Riemann--Hilbert problem for the $\q$--BOPs}
We will be describing here only the RHP characterizing the polynomials $q_n(y)$, where the characterization of the polynomials $p_n(x)$ is obtained by simply interchanging $\alpha$ with $\beta$ (see for example Theorem \ref{thm:Padep}).  

We consider the real axis $\R$ oriented as usual and define 
\bea
\vec \q^{(n)}_0(w):= \le[
\begin{array}{ccc}
q_{n-2}(w) \  
q_{n-1}(w)\  
q_{n}(w)
\end{array}\ri]^t, \quad 
\vec\q_{_1}^{(n)} (w):= \int \vec \q^{(n)}(y) \frac{\d\beta(y)}{w-y},
\quad \vec \q_{_2}^{(n)}(w) :=\int \vec \q_1^{(n)}(x)\frac{\d\alpha^*(x)}{w-x} 
\eea
For simplicity of notation we will suppress the superscript $^{(n)}$ in most of the following discussions, only to restore it when necessary for clarity; the main point is that an arrow on top of the corresponding vector will denote a ``window'' of three consecutive entries of either the ordinary vector $\q$ (index $a=0$), or the auxiliary vectors $\q_{_a}$ (index $a=1, 2$, see Def. \ref{defauxwave}) which, as we might recall at this point, 
combine the polynomials and the corresponding remainders in the 
Hermite-Pad\'{e} approximation problem given by Theorem \ref{thm:Padeq}.   
Some simple observations are in order.  
The vector $\vec\q_{_1}(w)$ is an analytic vector which  has a jump--discontinuity on the support of $\d\beta$ contained in the positive real axis. As $w\to \infty$ (away from the support of $\d\beta$) it decays as $\frac 1 w$. Its jump-discontinuity is (using Plemelj formula)
\be
\vec\q_{_1}(w)_+ = \vec\q_{_1}(w)_- - 2\pi i \frac{\d\beta}{\d w} \vec\q_{_0}(w)\ ,\ \ w\in supp(\d \beta).  
\ee
Looking at the leading term at $w=\infty$ we see that 
\be
\vec \q_{_1}(w)  =\frac{1}{w} \begin{bmatrix}
\eta_{n-2}& 
\eta_{n-1}&
\eta_{n}
\end{bmatrix}^t+ \mathcal O(1/w^{2})\ .
\ee
The vector  $\vec \q_{_2}(w)$ is also analytic with a jump discontinuity on the {\bf reflected support} of $\d\alpha$ (i.e. on $supp (\d \alpha^*)$). In view of Theorem \ref{thm:Padeq}, recalling that $\q_2$ are remainders 
of the Hermite-Pad\`{e} approximation problem of type II, 
we easily see that 
\bea
\vec\q_{_2}(w)  = \begin{bmatrix}
\ds \frac {c_{n-2}} {(-w)^{n-1}}  &
\ds \frac {c_{n-1}} {(-w)^{n}} &
\ds \frac {c_n} {(-w)^{n+1}} 
\end{bmatrix}^t
(1+\mathcal O(1/w)), \qquad 
c_n  := \langle x^n|q_n\rangle = \sqrt{\frac {D_{n+1}}{D_n}} >0.  
\label{q2asym}
\eea
The jump-discontinuity of $\vec \q_{_2}$ is 
\be
\vec\q_{_2}(w)_+ = \vec\q_{_2}(w)_-  -  2\pi i \frac{\d \alpha^*}{\d w}\vec \q_{_1}(w)\ \ \ \ w \in supp (d\alpha^*).  
\ee
The behavior of $\vec\q_{_0}(w)$ at infinity is 
\be
\vec\q_{_0}(w)  = 
\begin{bmatrix}
\ds \frac{w^{n-2}} {c_{n-2}}  
&
\ds\frac {w^{n-1}}{c_{n-1}}  &
\ds \frac {w^n}{c_n} 
\end{bmatrix}^t(1+\mathcal O(1/w)),
\ee
with the same $c_n$'s as in \ref{q2asym}. 

Define the matrix
\be
\Gamma(w) :=\overbrace{\begin{bmatrix}
1&-c_n \eta_n&0\\
0&1&0\\
0&(-1)^{n-1}\frac{\eta_{n-2}}{c_{n-2}}&1 \end{bmatrix}
\le[\begin{array}{ccc}
0 &0&c_n
\cr 
0 &\frac{1}{\eta_{n-1}}&0
\cr
\frac {(-1)^{n}} {c_{n-2}}&0&0
\end{array}\ri] }^{=: \mathcal N_q} [\vec \q_{_0}^{(n)}(w), \vec \q_{_1}^{(n)}(w), \vec \q_{_2}^{(n)}(w)]
\label{normalizedqRHP}
\ee

\bp
\label{RHP1}
The matrix $\Gamma(w)$ is analytic on $\C \setminus (supp(\d\beta)\cup 
supp(\d \alpha^*)$.  Moreover, it  satisfies the jump conditions
\be \label{eq:RHq}
\begin{split}
\Gamma(w)_+ & = \Gamma(w)_-\le[\begin{array}{ccc}
1 &  -2\pi i \frac{\d \beta}{\d w} & 0 \cr
0&1&0\cr
0&0&1
\end{array}\ri]\ , \qquad w\in supp(\d\beta)\subset \R_+\cr
\Gamma(w)_+ & = \Gamma(w)_- \le[
\begin{array}{ccc}
1&0&0\\
0&1&  -2\pi i \frac{\d \alpha^*}{\d w}\\
0&0&1
\end{array}
\ri]\ ,\qquad w\in supp(\d \alpha^*)\subset \R_-  
\end{split}
\ee
and its asymptotic behavior at $w=\infty$ is
\bea \label{eq:Gamma-as}
\Gamma(w)  = (\1 + \mathcal O(w^{-1}))\le[
\begin{array}{ccc}
w^n& 0 & \\
0& w^{-1}& 0 \\
0&0&w^{-n+1}
\end{array}
\ri]
\eea
Moreover, $\Gamma(w)$ can be written as: 
\begin{align}\label{eq:q-recovery}
\Gamma(w)=\begin{bmatrix}c_n\eta_n&0&0\\0&\frac{1}{\eta_{n-1}}&0\\0&0&\frac{(-1)^{n-1}\eta_{n-2}}{c_{n-2}}\end{bmatrix}
\begin{bmatrix}\wh q_{n-1}& \wh q_{1,n-1}&\wh q_{2,n-1}\\
q_{n-1}&q_{1,n-1}&q_{2,n-1}\\\wh q_{n-2}& \wh q_{1,n-2}&\wh q_{2,n-2}\end{bmatrix}. \end{align}
 
\ep
\begin{proof}
All the properties listed are obtained from elementary matrix computations.  
 
\end{proof}
\br 
An analogous problem with the r\^oles of $\alpha,\beta$, etc., interchanged, characterizes the monic orthogonal polynomials $p_{n-1}(x)$  of degree $n-1$ in $x$. 
\er 
\bc \label{cor:RHq}
Given  $n\in \mathbb{N}$, the absolutely continuous measures $\d\beta\subset \mathbb{R}_+$ and $\d \alpha^* \subset \mathbb{R}_-$, and assuming the existence of all the bimoments 
$I_{ij}$ there exists a unique 
matrix $\Gamma(w)$ solving the RHP specified by equations 
\eqref{eq:RHq}, \eqref{eq:Gamma-as}.  The solution characterizes 
uniquely the polynomials $q_{n-1}$ as well as $\wh q_{n-1}$.  In particular, 
 the normalization constants $c_{n-1},\eta_{n-1}$ (i.e. the ``norm'' of the monic orthogonal polynomials and the $\beta$ average of the $q_{n-1}$) are read off the following expansions
\bea
\Gamma_{2,1}(w)  = \frac 1{c_{n-1}\eta_{n-1}} w^{n-1} + \mathcal O(w^{n-2}),\qquad 
\Gamma_{2,3}(w) = (-1)^{n} \frac {c_{n-1}}{\eta_{n-1}w^n} + \mathcal O(w^{-n-1})
\eea
or, equivalently,  
\be
\frac{1}{\eta_{n-1}^2} =(-1)^n \lim_{w\to \infty} w\Gamma_{2,1}(w)\Gamma_{2,3}(w),\qquad 
c_{n-1}^2=(-1)^n\lim_{w\to \infty} w^{2n-1}\frac{\Gamma_{2,3}(w)}{\Gamma_{2,1}(w)}.  
\ee 
\ec
\begin{proof}
Given $\d\beta$ and $\d\alpha^*$ it suffices 
to construct the Nikishin systems $W_{\beta}, W_{\beta \alpha^*}$ and 
$W_{\alpha^*}, W_{\alpha^*\beta}$ followed by solving the Hermite-Pad\'{e} 
approximation problems given by equations \eqref{eq:PadeA}. The existence of the solution is ensured by the existence of all bimoments $I_{ij}$ (see equation 
\eqref{eq:bimoments} for the definition).   Then one constructs 
the polynomials $\wh q_{j}$, finally the matrix $\Gamma(w)$ using 
equation \eqref{eq:q-recovery}.  By construction $\Gamma(w)$ satisfies the 
Riemann-Hilbert factorization problem specified by equations \eqref{eq:RHq} and 
\eqref{eq:Gamma-as}.  
Since the determinant of $\Gamma(w)$ is constant in $w$ (and equal to one), the solution of the Riemann--Hilbert problem is unique.  The formulas for 
$\eta_{n-1}$ and $c_{n-1}$ follow by elementary matrix computations.  
\end{proof}

%

%
%%
%\subsubsection{A Riemann--Hilbert problem with constant jumps}
%%
%Let us recall that $
%\frac{\d \alpha}{\d x}  = {\rm e}^{-\frac 1\hbar U(x)} \ ,x\in supp(\d \alpha), \quad 
%\frac{\d \beta}{\d y} = {\rm e}^{-\frac 1 \hbar V(y)}\ ,y\in supp(\d \beta).  
%$
%In order to modify the RHP into one with constant jumps we make the (restrictive) assumption (only for this subsection) that the {\it potentials} can be extended to analytic functions off the real axis; for example, by assuming that $U(x),V(y)$ are real--analytic functions.  
%The matrix
%\be
%\YY (w) := \Gamma(w) \le[\begin{array}{ccc}
%\exp \le({-\frac {2V+U^\star}{3\hbar}}\ri)  & 0 & 0 \\
%0& \exp\le({\frac {V-U^\star}{3\hbar}}\ri)& 0 \\
%0&0&\exp \le({\frac{2U^\star+V}{3\hbar}}\ri)
%\end{array}\ri]\label{nonreg}
%\ee
%solves a similar RH problem but with {\bf constant} jump-discontinuity (and still with unit determinant); hence one can conclude immediately that it solves a linear ODE in the complex plane (or in the maximal domain of meromorphicity of $U,V$). The detailed singularity structure of this ODE depends on the analyticity properties of the potentials but the issue is of no relevance for the moment. For similar situation in the ordinary OP case see \cite{semiiso}.

%This new RHP exhibits the following jumps
%$$
%\YY_+(w) = \YY_-(w) \le[\begin{matrix}1&-2\pi i&0\cr
%0&1&0\cr
%0&0&1\end{matrix}\ri] \ ,\ w\in supp(\d\beta), \quad 
%\YY_+(w) = \YY_-(w) \le[\begin{matrix}1&0&0\cr
%0&1&-2\pi i \cr
%0&0&1\end{matrix}\ri] \ ,\ w\in supp(\d \alpha^*). 
%$$
\br
By multiplication on the right with a diagonal matrix 
$ \YY (w) := \Gamma(w) {\rm diag}\bigg(\exp \le({-\frac {2V+U^\star}{3\hbar}}\ri),\\ \exp\le({\frac {V-U^\star}{3\hbar}}\ri),\exp \le({\frac{2U^\star+V}{3\hbar}}\ri)\bigg)$
one can reduce the RHP to an equivalent one with constant jumps. It then follows that $\YY(w)$ solves a linear ODE with the same singularities as $V', {U^\star}'$; for example if $U', V'$ are rational functions then so is the coefficient matrix of the ODE and the orders of poles do not exceed those of $V', U'$.  In this case it can be shown \cite{Bertola:MomentTau} that the principal minors of the matrix of bimoments are  isomonodromic tau--functions in the sense of Jimbo--Miwa--Ueno \cite{JMU}.
%In a certain sense this is to be expected {\it a priori} because the vanishing of the isomonodromic tau function characterizes the non--solvability of the Riemann--Hilbert problem, i.e. the (non)--existence of the BOPs, exactly as the principal minors of $I$ do.
\er 
\subsection{ Riemann--Hilbert problem for the $\wh \p$--BOPs}
Referring to the defining properties of $\wh p_n(x)$ as indicated in Prop. \ref{hattedBOPs} we are going to define a second  $3\times 3$ local RHP that characterizes them.

Define
\be
\vec{\wh \p_{0}}(z) := \le[\begin{array}{ccc}
\wh p_{n-2}(z) &
\wh p_{n-1}(z) &
\wh p_{n}(z)
\end{array}\ri]^t
\ee
and $\vec {\wh \p}_{1,2}(z)$ as the same {\bf windows} 
of the auxiliary vectors $\wh \p _{1,2}$  introduced in Definition \ref{defauxwave}. 
We first study the large $z$ asymptotic behavior of 
$\wh p_{0,n}(z), \wh p _{1,n}(z), \wh p_{2,n}(z)$.  
\bl 
The asymptotic behavior at $z \to \infty, z\in \C_{\pm}$ is given by: 
\begin{align}\label{eq:phat-as}
&\wh p_{0,n}(z)=-\frac{\eta_n}{c_n}z^n(1+\mathcal{O}(1/z)),\\
&\wh p_{1,n}(z)=-1+ \mathcal{O}(1/z), \\ 
&\wh p_{2,n}(z)=(-1)^n \frac{c_{n+1} \eta_{n+1}}{z^{n+2}}(1+
\mathcal{O}(1/z)).  
\end{align}
\el 
\begin{proof}
We give a proof for $\wh p_{1,n}(z)=\int \frac{\wh p_{0,n}(x)}{z-x}\d\alpha(x)+
\frac{1}{\beta_0}\langle \wh p_{0,n}|1\rangle $.  The first term is $\mathcal{O}(\frac{1}{z})$, 
while the second term can be computed using biorthogonality and the fact that 
$\wh p_{0,n}=-(\eta_n p_{0,n}+\eta_{n-1}p_{0,n-1}+\cdots+ \eta_0p_{0,0})$.  
Thus the second term equals $-\frac{\eta_0}{\beta_0}\langle p_{0,0}|1\rangle =-1$, since 
$\eta_0=q_0 \beta_0$, hence the claim for $\wh p_{1,n}(z)$ follows.  
The remaining statements are proved in a similar manner.  
\end{proof}

For reasons of normalization, and in full analogy with equation 
\eqref{normalizedqRHP}, we 
arrange the window of all $\wh \p$s wave vectors into the matrix 
\be
\wh \Gamma(z) =\overbrace{
 \le[
\begin{array}{ccc}
0&0& -\frac {c_n}{\eta_n} \\
0&-1&0\\
\frac {(-1)^{n}}{c_{n-1}\eta_{n-1}} &0&0
\end{array}
\ri]\le[
\begin{array}{crc}
1 & -1 & 0\\
0 &1&0\\
0&-1 &1
\end{array}
\ri] }^{=: \mathcal N_{\wh p}} \le[\vec{\wh \p}(z),\vec{\wh \p}_1(z),\vec{\wh \p}_2(z)\ri].
\label{normalizedphatRHP}
\ee 
\bp
\label{RHP2}
The matrix $\wh \Gamma(z)$ is analytic in 
$\C \setminus supp(\d\alpha)\cup supp(\d\beta^*)$.  Moreover, it 
satisfies the jump conditions 
\be \label{eq:RHphat}
\begin{split}
\wh \Gamma(z) _+ &= \wh \Gamma(z)_- \le[
\begin{array}{ccc}
1 &   -2\pi i \frac{\d \alpha}{\d z}  & 0\cr
0&1&0\\
0&0&1
\end{array}
\ri]\ ,\ \ z\in supp(\d\alpha)\subseteq \R_+\\
\wh \Gamma(z) _+ &= \wh \Gamma(z)_- \le[
\begin{array}{ccc}
1 &  0 & 0\cr
0&1& -2\pi i \frac{\d \beta^*}{\d z}\\
0&0&1
\end{array}
\ri]\ ,\ \ z\in supp(\d \beta^*)\subseteq \R_-, 
\end{split}
\ee 
and its asymptotic behavior at $z=\infty$ is
\bea \label{eq:Gammahat-as}
\wh \Gamma(z) =  \le(\1  + \mathcal O\le(\frac 1 z\ri)\ri)\le[
\begin{array}{ccc}
z^n & 0& \cr
0 &1 &0\cr
0 &0& \frac 1{ z^{n}} 
\end{array} 
\ri].  
\eea
$\wh \Gamma(z)$ can be written as: 
\begin{align}\label{eq:p-recovery}
\Gamma(z)=\begin{bmatrix}c_n&0&0\\0&-1&0\\0&0&\frac{(-1)^n}{c_{n-1}}\end{bmatrix}
\begin{bmatrix}p_{0,n}&p_{1,n}&p_{2,n}\\
\wh p_{0,n-1}&\wh p_{1,n-1}&\wh p_{2,n-1}\\p_{0,n-1}& p_{1,n-1}&p_{2,n-1}\end{bmatrix}. \end{align}

\ep

The existence and uniqueness of the solution of the Riemann-Hilbert problem 
\eqref{eq:RHphat}, \eqref{eq:Gammahat-as} is proved 
in a similar way to the proof of Corollary \ref{cor:RHq}.  
\bc
Given  $n\in \mathbb{N}$, the absolutely continuous measures $\d\alpha\subset \mathbb{R}_+$ and $\d \beta^* \subset \mathbb{R}_-$, and assuming the existence of all the bimoments 
$I_{ij}$ there exists a unique 
matrix $\Gamma(z)$ solving the RHP specified by equations 
\eqref{eq:RHphat}, \eqref{eq:Gammahat-as}.  The solution characterizes 
uniquely the polynomials $\wh p_{n-1}$ and $p_{n}$.  
\ec
\section{ Acknowledgments}

M.B. would like to thank the Department of Mathematics of the 
University of Notre Dame for hospitality during which the project was initiated and J. Harnad for insight on the relationship of Cauchy biorthogonal polynomials with matrix models. 

While working on this project, M. B. and M. G. enjoyed  the hospitality of
the Department of Mathematics, University of Saskatchewan and M. G.  and J. S.
 enjoyed  the  hospitality of the Centre de recherches math\'ematiques,
Universit\'e de Montr\'eal.  

J.S would also like to thank 
H. Lundmark for an ongoing collaboration on the cubic string problem which 
motivated many of the questions  addressed in this paper.    

\appendix
\section{Appendix: Proof of Extended Christoffel-Darboux Identities}
\label{sec:app-CD}
\bt [Extended Christoffel-Darboux Identities \ref{thm:ECD1}]
Let ${a,b=0,\dots 2}$.  Then  
\be
(w+z) \q_{_a}^T(w) \Pi  \p_{_b}(z) = \q_{_a}^T(w) \mathbb A(-w) \wh \p_{_b}(z)-\mathbb F(w,z)_{ab}
\ee
where 
\be
\mathbb F(w,z)= \begin{bmatrix}
0&0&1\\
0& 1&W_{\beta^*}(z)+W_{\beta}(w) \\
1&W_{\alpha}(z)+W_{\alpha^*}(w)&W_{\alpha^*}(w)W_{\beta^*}(z)+W_{\alpha^*\beta}(w)+W_{\beta^*\alpha}(z)  
\end{bmatrix}.  
\ee 

\et 
\begin{proof}
The proof goes by repeated applications of the Christoffel-Darboux Identities 
given by Theorem \ref{thm:CD1} and Pad\'{e} approximation conditions 
\eqref{eq:PadeA}.  We observe that all quantities with labels $a=1,2$ 
have asymptotic expansions around $\infty $ in the open 
half-planes $\C_{\pm}$ (they are holomorphic expansions in the case of compactly supported measures $\d\alpha, \d\beta$).  We 
will subsequently call the part of the expansion corresponding to 
negative powers of $z$ or $w$,  of a function $f(z,w)$ the { \sl regular }
part of $f$ and denote it $\le( f(z,w)\ri)_{-,z} $, 
$\le( f(z,w)\ri)_{-,w} $ respectively.   In all cases the {\sl regular} 
parts are obtained by subtracting certain polynomial expressions 
from functions holomorphic in $\C_{\pm}$ and as such 
the {\sl regular} parts are holomorphic in these half-planes  with vanishing 
limits at $\infty$ approach from within the respective half-planes.  

We will only indicate the main steps in computations for each entry, 
denoted below by $(a,b)$.  
 
\noindent (0,1):  With the help of the first 
approximation condition, we have
$$
 \q_1^T(w) \Pi  \p_{0}(z)=\le (\int \frac{\q_0^T(w)\Pi \p _0(z)}{w-y}\d \beta(y)\ri)_{-,w}. 
$$
Using the Christoffel-Darboux Identities and the notation of Corollary 
\ref{cor:CDIuni} we get 
\begin{align*}
& \q_1^T(w) \Pi  \p_{0}(z)=\le(\int \frac{\q_0^T(w)\mathbb{A}(-w)
\wh \p _0(z)}{(w+z)(w-y)}\d \beta(y)\ri)_{-,w}=\\
&\int \frac{\q_0^T(y)\mathbb{A}(-w)
\wh \p _0(z)}{(w+z)(w-y)}\d \beta(y)+\le(\int \frac{(\q_0^T(w)-\q_0^T(y))\mathbb{A}(-w)
\wh \p _0(z)}{(w+z)(w-y)}\d \beta(y)\ri)_{-, w},
\end{align*}
where we dropped the projection sign in the first term because $\mathbb{A}(-w)$ 
is a polynomial of degree one. Using now the partial fraction decomposition 
$$
\frac{1}{(w+z)(w-y)}= \frac{1}{z+y}\le(\frac{1}{w-y}-\frac{1}{w+z}\ri),
$$ 
we get that 
$$
\le(\int \frac{(\q_0^T(w)-\q_0^T(y))\mathbb{A}(-w)
\wh \p _0(z)}{(w+z)(w-y)}\d \beta(y)\ri)_{-,w}=-\le(\int \frac{(\q_0^T(-z)-\q_0^T(y))
[\Pi, (-z-\Y ^T)\wh L]
\wh \p _0(z)}{(w+z)(z+y)}\d \beta(y)\ri)_{-,w}. 
$$
 
Observe that $(-z-\Y ^T)\wh L
\wh \p _0(z)=0$, $\q_0^T(-z)(-z-\Y ^T)\wh L=0$ and $ \q_0^T(y)(-z-\Y^T)\wh L=
-(y+z)\q_0^T(y)\wh L$ so
\begin{align*}
&\le (\int \frac{(\q_0^T(w)-\q_0^T(y))\mathbb{A}(-w)
\wh \p _0(z)}{(w+z)(w-y)}\d \beta(y)\ri )_{-,w}=\le(\int \frac{(\q_0^T(y))(z+\Y^T)\wh L\Pi
\wh \p _0(z)}{(w+z)(z+y)}\d \beta(y)\ri)_{-,w} =\\
&\int \frac{\q_0^T(y)\wh L\Pi
\wh \p _0(z)}{w+z}\d \beta(y)=0, 
\end{align*}
because the $\beta$ averages of $\wh \q$ are zero.  Thus
$$
(w+z) \q_1^T(w) \Pi  \p_{0}(z)=\q_1^T(w)\mathbb A(-w)\wh \p_0(z). 
$$ 

\noindent (2,0):  Using the second Pad\`{e} approximation condition and biorthogonality we easily obtain 
$$
\mathbf R^T_{\beta \alpha^*}(w)\Pi \p_0(z)=\frac{\mathbf R^T_{\beta \alpha^*}(w)\mathbb A(-w)\wh \p_0(z)+1}{w+z},   
$$

Now, substituting this formula into the formula for the third approximation condition, written as in equation \eqref{eq:3rdPadeA}, gives: 
$$
\mathbf R^T_{\alpha^*\beta}(w)\Pi \p_0(z)=
\frac{\mathbf R^T_{\alpha^*\beta}(w)\mathbb A(-w)\wh \p_0(z)-1}{w+z}. 
$$
Restoring the collective notation of $\q _a, \p_a$ we obtain :
$$
(w+z)\q_2^T(w)\Pi \p_0(z)=\q_2^T(w)\mathbb A(-w)\wh \p_0(z)-1.  
$$
\noindent (0,1):  To compute $\q_0^T(w)\Pi \p_1(z)$ we use the Pad\`{e} approximation conditions 
\eqref{eq:PadeB}, in particular the first condition gives us: 
$$
\q_0^T(w)\Pi\p_0(z) W_{\alpha}(z)-\q_0^T(w)\Pi\mathbf{P}_{\alpha}(z)=
\q_0^T(w)\Pi\mathbf{R}_{\alpha}(z).  
$$

We observe that this time we have to project on the negative powers of 
$z$.   Thus the goal is to compute  $\le (\q_0^T(w)\Pi\p_0(z) W_{\alpha}(z)\ri)_{-,z}$.  
We have 
\begin{align*}
\le(\int \frac{\q_0^T(w)\Pi\p_0(z) \d \alpha(x)}{z-x}\ri)_{-,z}=
\le(\int \frac{\q_0^T(w)\mathbb A(-w)\wh\p_0(z) \d \alpha(x)}{(z-x)(w+z)}\ri)_{-,z}
=\\
\le(\int \frac{\q_0^T(w)\mathbb A(-w)\wh \p_0(x) \d \alpha(x)}{(z-x)(w+z)}\ri)_{-,z}
+\le(\int \frac{\q_0^T(w)\mathbb A(-w)(\wh\p_0(z)-\wh\p_0(x)) \d \alpha(x)}
{(z-x)(w+z)}\ri)_{-,z}.  
\end{align*}
We see that the first term is already {\sl regular}  in $z$.  
To treat the second term we perform the partial fraction expansion
$\frac{1}{(z-x)(w+z)}=\frac{1}{w+x}[{\frac{1}{z-x}-\frac{1}{w+z}}]$ and 
observe that the term with $\frac{1}{z-x}$ does not contribute, while the 
second term
\begin{align*}
&-\le(\int \frac{\q_0^T(w)\mathbb A(-w)(\wh\p_0(z)-\wh\p_0(x)) \d \alpha(x)}
{(w+x)(w+z)}\ri)_{-,z}=
-\le(\int \frac{\q_0^T(w)\mathbb A(-w)(\wh\p_0(-w)-\wh\p_0(x)) \d \alpha(x)}
{(w+x)(w+z)}\ri)_{-,z}=\\
&\int \frac{\q_0^T(w)\mathbb A(-w)\wh\p_0(x) \d \alpha(x)}
{(w+x)(w+z)}.  
\end{align*}

Thus 
$$
\q_0^T(w)\Pi\p_1(z)=\frac{\q_0^T(w)\mathbb A(-w)\wh L ^{-1} \p_1(z)}{w+z}-\frac{\q_0^T(w)\mathbb A(-w)\wh L ^{-1} \p_1(-w)}{w+z}.  
$$

In other words, 
$$
(w+z)\q_0^T(w)\Pi\p_1(z)=\q_0^T(w)\mathbb A(-w)\wh L ^{-1}(\p_1(z)-\p_1(-w)).  
$$
 More explicitly, the second term above 
can be rewritten as 
$$
-\q_0^T(w)\mathbb A(-w)\wh L^{-1} \p_1(-w)=\q_0^T(w)\Pi\int \p(x)\d\alpha(x).  
$$
On the other hand 
\begin{align*}
&\q_0^T(w)\mathbb A(-w)\iint \frac{\wh \p(x)\d \alpha(x)\d \beta(y)}
{\beta_0 (x+y)}=
\q_0^T(w)\Pi \iint \frac{(w+x)\p(x)\d \alpha(x)\d \beta(y)}
{\beta_0 (x+y)}=\\
&\q_0^T(w)\Pi\int \p(x)\d\alpha(x)+\q_0^T(w)\Pi \iint \frac{(w-y)\p(x)\d \alpha(x)\d \beta(y)}
{\beta_0 (x+y)}. 
\end{align*} 
Now the second term 
$\q_0^T(w)\Pi \iint \frac{(w-y)\p(x)\d \alpha(x)\d \beta(y)}{\beta_0 (x+y)}=0$ because $\q_0^T(w)\Pi \langle \p(x)|\bullet \rangle$ is a projector 
on polynomials of degree $\leq n-1$ and thus $w\q_0^T(w)\Pi\langle\p(x)|1\rangle-
\q_0^T(w)\Pi\langle\p(x)|y\rangle=w-w=0$, hence 
$$
(w+z)\q_0^T(w)\Pi\p_1(z)=\q_0^T(w)\mathbb A(-w)\wh \p_1(z),  
$$
where $\wh \p_1(z)=\wh L^{-1}  (\p_{_1}(z)+\frac{1}{\beta_0}\langle\p|1\rangle)$ as 
advertised earlier.  

\noindent $(1,1)$:  We use again the Pad\`{e} approximation conditions 
\eqref{eq:PadeB}, this time multiplying on the left by $\q_1^T(w)\Pi$
and projecting on the  negative powers 
of $z$ , to obtain: 
$$
\le(\q_1^T(w)\Pi\p_0(z) W_{\alpha}(z)\ri)_{-,z}=
\q_1^T(w)\Pi\p _1(z).  
$$
With the help of the result for the $(0,1)$ entry, after carrying out the 
projection, we obtain
$$
(w+z)\q_1^T(w)\Pi\p _1(z)=\q_1^T(w)\mathbb A(-w)\wh \p_1(z)+
\q_1^T(w)\mathbb A(-w)\le(\int \frac{\wh \p(x) \d \alpha(x)}{w+x}-
\frac{1}{\beta_0}\langle\wh \p|1\rangle\ri). 
$$ 
We claim that 
$$
\q_1^T(w)\mathbb A(-w)\le(\int \frac{\wh \p(x) \d \alpha(x)}{w+x}-
\frac{1}{\beta_0}\langle\wh \p|1\rangle\ri)=-1. 
$$
Indeed, the left hand side of the equation equals:
\begin{align*}
&\frac{1}{\beta_0}\q_1^T(w)\Pi\iint \frac{(y-w)\p(x)\d\alpha(x) \d \beta(y)}{x+y}=\frac{1}{\beta_0}\int \frac{\q _0^T(\xi)}{w-\xi}\Pi (\langle \p |y\rangle-w\langle \p|1\rangle)\d \beta(\xi)=\\
&\frac{1}{\beta_0}\int \frac{\xi-w}{w-\xi}\d \beta(\xi)=-1.  
\end{align*} 
 Thus  
$$
(w+z)\q_1^T(w)\Pi\p _1(z)=\q_1^T(w)\mathbb A(-w)\wh \p_1(z)-1. 
$$

\noindent $(2,1)$:  This time we use projections 
in both variables, one at a time, and compare the results. 
First, let us use the projections in $z$. Thus 

$$
  \q_2^T(w)\Pi \p_1(z)=\le(\q_2^T(w)\Pi \p_0(z) W_{\alpha}(z)\ri)_{-,z}. 
$$
Carrying out all the projections we obtain an expression of the form: 
$$
\q_2^T(w)\Pi \p_1(z)=\frac{\q_2^T(w)\mathbb A(-w)\wh \p_1(z)}{w+z}-
\frac{W_{\alpha}(z)+F(w)}{w+z}. 
$$
Observe that, since $\q_2^T(w)$ is $\mathcal {O}(1/w)$ and the first term on the right is much smaller, $F(w)=\mathcal{O}(1)$.  More precisely, 
by comparing the terms at $1/w$ on both sides, we conclude that in fact, 
$F(w)=\mathcal{O}(1/w)$.  
Now, we turn to the projection in $w$, resulting in an expression 
of the form:
$$
\q_2^T(w)\Pi \p_1(z)=\frac{\q_2^T(w)\mathbb A(-w)\wh \p_1(z)}{w+z}-
\frac{W_{\alpha^*}(w)+G(z)}{w+z}.$$ 
This, and the fact that $F(w)=\mathcal{O}(1/w)$, implies that 
$F(w)=W_{\alpha^*}(w), G(z)=W_{\alpha}(z)$.  Hence 
$$
(w+z)\q_2^T(w)\Pi \p_1(z)=\q_2^T(w)\mathbb A(-w)\wh \p_1(z)-(W_{\alpha}(z)+W_{\alpha^*}(w)). 
$$

\noindent $(0,2)$:  We use the projection in the $z$ variable and the fact that by the Pad\`{e} approximation condition 
\eqref{eq:3rdPadeA}, after exchanging $\alpha$ with $\beta$,
$\p_2(z)=\p_1(z)W_{\beta^*}(z)-\mathbf R _{\alpha\beta^*}(z)$.  
Using the result for the $(0,1)$ entry we obtain:
$$
\q_0^T(w)\Pi\p_2(z)=\frac{\q_0^T(w)\mathbb A(-w)\p_1(z)W_{\beta^*}(z)}{w+z} -\le(\frac{\q_0^T(w)\mathbb A(-w)\p_0(z)W_{\alpha\beta^*}(z)}{w+z}\ri)_{-,z}. 
$$
Carrying out the projection and reassembling terms according to the definition 
of $\wh \p_2(z)$ we obtain: 
$$
\q_0^T(w)\Pi\p_2(z)=\frac{\q_0^T(w)\mathbb A(-w)\wh \p_2(z)}{w+z}
-\frac{\q_0^T(w)\Pi \langle \p_0|1\rangle}{w+z}=\frac{\q_0^T(w)\mathbb A(-w)\wh \p_2(z)}{w+z}-\frac{1}{w+z}.  
$$

\noindent $(1,2)$:  We use the projection in the $z$ variable
and the Pad\`{e} approximation condition $\p_2(z)=\p_1(z)W_{\beta^*}(z)-\mathbf R _{\alpha\beta^*}(z)$.  Consequently, 
\begin{align*}
&\q_1^T(w)\Pi \p_2(z)=\q_1^T(w)\Pi \p_1(z)W_{\beta^*}(z)-\q_1^T(w)\Pi\mathbf R_{\alpha \beta^*}(z)=\\
&\le(\frac {\q_1^T(w)\mathbb A(-w)\wh \p_1(z)-1}{w+z}\ri)W_{\beta^*(z)}-\le(\q_1^T(w)\Pi \p_0(z)W_{\alpha \beta^*}(z)\ri)_{-,z}.  
\end{align*}
Using the existing identities and carrying out the projection in the second 
term we obtain: 
$$
(w+z)\q_1^T(w)\Pi \p_2(z)=\q_1^T(w)\mathbb A(-w)\wh \p_2(z)-W_{\beta^*}(z)-W_{\beta}(w). 
$$

\noindent $(2,2)$:  The computation is similar to the one for $(1,2)$ entry; we use both 
projections.  
The projection in the $z$ variable gives: 
$$
\q_2^T(w)\Pi \p_2(z)=\frac{\q_2^T(w)\mathbb A(-w)\wh \p_2(z)}{w+z}
+\frac{F(w)-(W_{\alpha^*}(w)+W_{\alpha}(z))W_{\beta^*}(z)+W_{\alpha \beta^*}(z)}{w+z}.  
$$
On the other hand, carrying out the projection in the $w$ variable we obtain: 
$$
\q_2^T(w)\Pi \p_2(z)=\frac{\q_2^T(w)\mathbb A(-w)\wh \p_2(z)}{w+z}
+\frac{G(z)-(W_{\beta}(w)+W_{\beta^*}(z))W_{\alpha^*}(w)+W_{\beta\alpha^*}(w)}{w+z}.  
$$
Upon comparing the two expressions and using Lemma \ref{lem:Plucker} we 
obtain $F(w)=-W_{\alpha^* \beta}(w)$, hence 
\begin{align*}
&(w+z)\q_2^T(w)\Pi \p_2(z)=\q_2^T(w)\mathbb A(-w)\wh \p_2(z)
-W_{\alpha^* \beta}(w)-(W_{\alpha^*}(w)+W_{\alpha}(z))W_{\beta^*}(z)+W_{\alpha \beta^*}(z)=\\
&\q_2^T(w)[\Pi, \mathbb A(-w)]\wh \p_2(z)
-(W_{\alpha^*}(w)W_{\beta^*}(z)+W_{\alpha^*\beta}(w)+W_{\beta^*\alpha}(z)),  
\end{align*}
where in the last step we used again Lemma \ref{lem:Plucker}. 
  \end{proof}

\bibliographystyle{plain}
\bibliography{BOP}
\end{document}